%% file: arxiv.tex
\documentclass[a4paper, 10pt]{article}

\newif\ifalgrules\algrulestrue

\newif\ifjune\junefalse

\newif\ifappendix\appendixtrue

\usepackage{authblk}
\usepackage{amsmath}
\usepackage{amssymb}
\usepackage{amsthm}
\usepackage{color}
\usepackage{mathpartir}
\usepackage{multirow}
\usepackage{stmaryrd}
\usepackage{cmll}
\usepackage[numbers]{natbib}
\usepackage{verbatim}
\usepackage{thmtools}
\usepackage{thm-restate}
\usepackage{hyperref}
\hypersetup{colorlinks,allcolors=black}
\usepackage{cleveref}
\usepackage[a4paper, total={6.5in, 9.5in}]{geometry}
\usepackage [english]{babel}
\usepackage [autostyle, english = american]{csquotes}
\MakeOuterQuote{"}

\ifjune
\usepackage{fontspec}
\fi

\usepackage[margin=false,inline=true]{fixme}
\fxsetup{status=final}
\FXRegisterAuthor{aaa}{anaaa}{\color{cyan}AAA}
\FXRegisterAuthor{mg}{anmg}{\color{red}MG}
\FXRegisterAuthor{jw}{anjw}{\color{orange}JW}
\FXRegisterAuthor{pb}{anpb}{\color{blue}PB}
\newcommand{\aaa}[1]{\aaanote{#1}}

\newcommand{\pb}[1]{\pbnote{#1}}

\DeclareMathOperator*{\Set}{\mathit{set}}
\newcommand{\R}{\mathbb{R}}
\newcommand{\Sens}{\mathbb{R}^{\geq 0}_{\infty}} 
\newcommand{\Dist}{\mathbb{R}^{\geq 0}_{\infty}}
\newcommand{\db}{\mathsf{db}}

\newcommand{\RR}{\(\mathbb R\)R}
\newcommand{\lolliR}{\(\multimap\)R}
\newcommand{\lolliL}{\(\multimap\)L}
\newcommand{\tensorR}{\(\otimes\)R}
\newcommand{\tensorL}{\(\otimes\)L}

\newcommand{\plusR}[1]{\(\oplus_{#1}\)R}
\newcommand{\plusL}{\(\oplus\)L}

\newcommand{\RI}{\(\mathbb R\)I}
\newcommand{\lolliI}{\(\multimap\)I}
\newcommand{\lolliE}{\(\multimap\)E}
\newcommand{\tensorI}{\(\otimes\)I}
\newcommand{\tensorE}{\(\otimes\)E}
\newcommand{\plusI}[1]{\(\oplus_{#1}\)I}
\newcommand{\plusE}{\(\oplus\)E}

\newcommand{\CaseOf}[5]{\textbf{case } #1 \textbf{ of } #2.\;#3 \mid #4.\;#5}
\newcommand{\LetPairIn}[4]{\textbf{let } (#1, #2) = #3 \textbf{ in } #4}
\newcommand{\LetBangIn}[3]{\textbf{let } !#1 = #2 \textbf{ in } #3}
\newcommand{\mLet}[3]{\textbf{mlet } #1 = #2 \textbf{ in } #3}
\newcommand{\LetIn}[3]{\textbf{let } #1 = #2 \textbf{ in } #3}
\newcommand{\return}{\textbf{return }}
\newcommand{\inj}{\textbf{inj}}

\newcommand{\Met}{\mathsf{Met}}
\newcommand{\MD}{\mathsf{MD}}
\newcommand{\HD}{\mathsf{HD}}
\newcommand{\lift}[1]{#1^\dagger}
\newcommand{\lp}{$L^p\text{ }$}
\newcommand{\system}{Bunched Fuzz}

\DeclareMathOperator*{\arginf}{arginf}

\DeclareMathOperator*{\plist}{\mathit{p}\mathtt{list}}
\DeclareMathOperator*{\Fuzzlist}{\mathtt{list}}
\DeclareMathOperator*{\set}{\mathtt{set}}
\DeclareMathOperator*{\colim}{\mathrm{colim}}

\renewcommand{\!}{\:!}
\renewcommand{\,}{\:,}

\newcommand{\intrp}[1]{\llbracket #1 \rrbracket}

\newtheorem{theorem}{Theorem}[section]
\newtheorem{lemma}[theorem]{Lemma}
\newtheorem{proposition}[theorem]{Proposition}
\theoremstyle{definition}
\newtheorem{definition}[theorem]{Definition}
\theoremstyle{remark}
\newtheorem*{remark}{Remark}

\ifjune
\fi

\newcommand\doubleplus{+\kern-1.3ex+\kern0.8ex}
\newcommand\mdoubleplus{\ensuremath{\mathbin{+\mkern-10mu+}}}

\begin{document}


\title{\system{}: Sensitivity for Vector Metrics}
\author[1]{june wunder}
\author[1]{Arthur Azevedo de Amorim}
\author[2]{Patrick Baillot}
\author[1]{Marco Gaboardi}
\affil[1]{Boston University, USA}
\affil[2]{Univ. Lille, CNRS, Inria, Centrale Lille, UMR9189 CRIStAL, F-59000 Lille, France}

\maketitle

\begin{abstract}
  \emph{Program sensitivity} measures the distance between the outputs of a
  program when run on two related inputs.  This notion, which plays a key role
  in areas such as data privacy and optimization, has been the focus of several
  program analysis techniques introduced in recent years.  Among the most
  successful ones, we can highlight type systems inspired by linear logic, as
  pioneered by Reed and Pierce in the Fuzz programming language.  In Fuzz, each
  type is equipped with its own distance, and sensitivity analysis boils down to
  type checking.  In particular, Fuzz features two product types, corresponding
  to two different notions of distance: the \emph{tensor product} combines the
  distances of each component by \emph{adding} them, while the \emph{with
    product} takes their \emph{maximum}.

  In this work, we show that these products can be generalized to arbitrary
  \emph{$L^p$ distances}, metrics that are often used in privacy and
  optimization.  The original Fuzz products, tensor and with, correspond to the
  special cases $L^1$ and $L^\infty$.  To ease the handling of such products, we
  extend the Fuzz type system with \emph{bunches}---as in the logic of bunched
  implications---where the distances of different groups of variables can be
  combined using different $L^p$ distances.  We show that our extension can be
  used to reason about quantitative properties of probabilistic programs.
\end{abstract}

%
%
%
%
%
%
%
\section{Introduction}
When developing a data-driven application, we often need to analyze its
\emph{sensitivity}, or \emph{robustness}, a measure of how its outputs can be
affected by varying its inputs. For example, to analyze the privacy guarantees
of a program, we might consider what happens when we include the data of one
individual in its inputs~\cite{DBLP:conf/tcc/DworkMNS06}. When analyzing the
stability of a machine-learning algorithm, we might consider what happens when
we modify one sample in the training set~\cite{DBLP:journals/jmlr/BousquetE02}.

Such applications have spurred the development of several techniques to reason
about program sensitivity~\cite{DBLP:conf/icfp/ReedP10,
  DBLP:conf/sigsoft/ChaudhuriGLN11}. One successful approach is based on
linear-like~\cite{DBLP:journals/tcs/Girard87} type systems, as pioneered in Reed
and Pierce's \emph{Fuzz} language~\cite{DBLP:conf/icfp/ReedP10}.

The basic idea behind Fuzz is to use typing judgments to track the sensitivity
of a program with respect to each variable.  Each type comes equipped with a
notion of distance, and the typing rules explain how to update variable
sensitivities for each operation.  Because different distances yield different
sensitivity analyses, it is often useful to endow a set of values with different
distances, which leads to different Fuzz types.  For example, like linear logic,
Fuzz has two notions of products: the tensor product $\otimes$ and the Cartesian
product $\with$ (with).  The first one is equipped with the $L^1$ (or Manhattan)
distance, where the distance between two pairs is computed by \emph{adding} the
distances between the corresponding components.  The second one is equipped with
the $L^\infty$ (or Chebyshev) distance, where the component distances are
combined by taking their \emph{maximum}.

The reason for focusing on these two product types is that they play a key role
in differential privacy~\cite{DBLP:conf/tcc/DworkMNS06}, a rigorous notion of
privacy that was the motivating application behind the original Fuzz design.
However, we could also consider equipping pairs with more general \emph{$L^p$
  distances}, which interpolate between the $L^1$ and $L^\infty$ and are
extensively used in convex optimization~\cite{coBook04}, information
theory~\cite{CIT-004} and statistics~\cite{10.5555/64130}.  Indeed, other type
systems for differential privacy inspired by
Fuzz~\cite{DBLP:journals/pacmpl/NearDASGWSZSSS19} include types for vectors and
matrices under the $L^2$ distance, which are required to use the Gaussian
mechanism, one of the popular building blocks of differential privacy. Supporting more general \lp{} metrics would allow us to capture even more such building blocks~\cite{HardtT10,osti_10183971}, which would enable further exploration of the tradeoffs between differential privacy and accuracy.

In this paper, we extend these approaches and show that Fuzz can be enriched
with a family of tensor products $\otimes_p$, for $1\leq p \leq\infty$. These
tensor products are equipped with the $L^p$ distance, the original Fuzz products
$\otimes$ and $\with$ corresponding to the special cases $\otimes_1$ and
$\otimes_\infty$.  Moreover, each connective $\otimes_p$ is equipped with a
corresponding ``linear implication'' $\multimap_p$, unlike previous related
systems where such an implication only exists for $p = 1$.  Following prior
work~\cite{DBLP:conf/popl/AmorimGHKC17, DBLP:conf/lics/AmorimGHK19}, we give to
our extension a semantics in terms of non-expansive functions, except that the
presence of the implications $\multimap_p$ forces us to equip input and output
spaces with more general distances where the triangle inequality need not hold.

A novelty of our approach is that, to support the handling of such products, we
generalize Fuzz environments to \emph{bunches}, where each $L^p$ distance comes
with its own context former. Thus, we call our type system \system.
This system, inspired by languages derived from the logic of Bunched
Implications (BI)~\cite{DBLP:journals/bsl/OHearnP99}
(e.g.~\cite{DBLP:journals/jfp/OHearn03}), highlights differences between the
original Fuzz design and linear logic---for example, products distribute over
sums in Fuzz and BI, but not in linear logic.
While similar indexed products and function spaces have also appeared in the
literature, particularly in works on categorical
grammars~\cite{DBLP:series/lncs/6850}, here they are employed to reason about
vector distances and function sensitivity.

While designing \system, one of our goals was to use sensitivity to reason about
randomized algorithms.  In the original Fuzz, probability distributions are
equipped with the \emph{max divergence} distance, which can be used to state
differential privacy as a sensitivity
property~\cite{DBLP:conf/icfp/ReedP10}. Subsequent work has shown how Fuzz can
also accommodate other distances over probability
distributions~\cite{DBLP:conf/lics/AmorimGHK19}. However, such additions
required variants of \emph{graded monads}, which express the distance between
distributions using indices (i.e. grades) on the monadic type of distributions
over their \emph{results}, as opposed to sensitivity indices on their
\emph{inputs}, as it was done in the original Fuzz.  In particular, this makes
it more difficult to reason about distances separately with respect to each
input.  Thanks to bunches, however, we can incorporate these composition
principles more naturally. For example, \system{} can reason about the Hellinger
distance on distributions without the need for output grading, as was done in
prior systems~\cite{DBLP:conf/lics/AmorimGHK19}.

We will also see that, by allowing arbitrary \lp norms, we can generalize prior
case studies that were verified in Fuzz and obtain more general methods for
reasoning about differential privacy (\Cref{sect:examples}).  Consider the \lp
mechanism~\cite{osti_10183971,HardtT10}, which adds noise to the result of a
query whose sensitivity is measured in the \lp norm.  Since Fuzz does not have
the means to analyze such a sensitivity measure, it cannot implement the \lp
mechanism; Bunched Fuzz, however, can analyze such a measure, and thus allows
for a simple implementation in terms of the exponential mechanism.  Such a
mechanism, in turn, can be used to implement a variant of a gradient descent
algorithm that works under the \lp norm, generalizing an earlier version that
was biased towards the $L^1$ norm~\cite{DBLP:journals/pacmpl/Winograd-CortHR17}.
Summarizing, our contributions are:
\begin{itemize}
\item We introduce \system, an extension of Fuzz with types for general $L^p$
  distances: we add type constructors of the form $\otimes_p$ (for
  $1 \leq p \leq \infty$) for pairs under the $L^p$ distance along with
  constructors of the form $\multimap_p$ for their corresponding function
  spaces.  To support the handling of such types, we generalize Fuzz typing
  contexts to \emph{bunches of variable assignments}.

\item We give a denotational semantics for \system{} by interpreting programs as
  non-expansive functions over spaces built on $L^p$ distances.

\item We show that \system{} can support types for probability distributions for
  which the sampling primitive, which enables the composition of probabilistic
  programs, is compatible with $L^p$ distances.

\item We show a range of examples of programs that can be written in
  \system. Notably, we show that \system{} can support reasoning about the
  Hellinger distance without the need for grading, and we show generalizations of several examples from the differential privacy literature.

\end{itemize}
Check the full version of this paper for more technical
details~\cite{DBLP:journals/corr/abs-2202-01901}.
\section{Background}

\subsection{Metrics and Sensitivity}
\label{sec:metric-spaces}

To discuss sensitivity, we first need a notion of distance. We call
\textit{extended pseudosemimetric space} a pair $X=(|X|,d_X)$ consisting of a
carrier set $|X|$ and an \textit{extended pseudosemimetric}
$d_X : |X|^2\rightarrow \Dist$, which is a function satisfying, for all
$x, y \in |X|$:
\begin{enumerate}
\item  $d_X(x,x) = 0$,
\item  $d_X(x,y) = d_X(y,x)$.
\end{enumerate} This relaxes the
standard notion of metric space in a few respects. First, the distance between
two points can be infinite, hence the \emph{extended}. Second, different points
can be at distance zero, hence the \emph{pseudo}. Finally, we do not require the
\emph{triangular inequality}:
\begin{align} \label{triangulaarinequality}
 d_X(x,y) \leq d_X(x,z)+ d_X(z,y),
\end{align}
hence the \emph{semi}.
We focus on extended pseudosemimetrics because they support constructions that
true metrics do not. In particular, they make it possible to scale the distance
of a space by $\infty$ and enable more general function spaces.  However, to
simplify the terminology, we will drop the ``extended pseudosemi'' prefix in the
rest of the paper, and speak solely of metric spaces.  In some occasions, we
might speak of a \emph{proper metric space}, by which we mean a space where the
triangle inequality \emph{does} hold (but not necessarily the other two
requirements that are missing compared to the traditional definition of metric
space).

\aaa{We might end up adding $\square$ sensitivity to the type system to avoid
  issues with 0 and $\infty$.  If so, we should use $\Sens$ to denote the set of
  sensitivities and $\Dist$ to denote the set of distances.}

Given a function $f : X \to Y$ on metric spaces, we say that it is
$s$-sensitive, for $s$ in $\Dist$, if we have:
\[ \forall x_1, x_2 \in X, \; d_Y(f(x_1),f(x_2))\leq s \cdot d_X(x_1,x_2), \]
(We extend addition and multiplication to $\Dist$ by setting
$\infty \cdot s = s \cdot \infty = \infty$.)  We may also say that $f$ is
\emph{$s$-Lipschitz continuous}, though the traditional definition of Lipschitz
continuity does not include the case $s = \infty$. If a function is
$s$-sensitive, then it is also $s'$-sensitive for every $s' \geq s$.
%
Every function of type $X \to Y$ is $\infty$-sensitive.
If a function is 1-sensitive, we also say that $f$ is non-expansive. We use
$X \multimap Y$ to denote the set of such non-expansive functions.  The identity
function is always non-expansive, and non-expansive functions are closed under
composition.  Thus, metric spaces and non-expansive functions form a category,
denoted $\Met$.

\subsection{Distances for Differential Privacy}
\label{sec:dp}

Among many applications, sensitivity is a useful notion because it provides a
convenient language for analyzing the privacy guarantees of
algorithms---specifically, in the framework of differential
privacy~\cite{DBLP:conf/tcc/DworkMNS06}.  Differential privacy is a technique
for protecting the privacy of individuals in a database by blurring the results
of a query to the database with random noise.  The noise is calibrated so that
each individual has a small influence on the probability of observing each
outcome (while ideally guaranteeing that the result of the query is still
useful).

Formally, suppose that we have some set of databases $\db$ equipped with a
metric.  This metric roughly measures how many rows differ between two
databases, though the exact definition can vary.  Let $f : \db \to DX$ be a
randomized database query, which maps a database to a discrete probability
distribution over the set of outcomes $X$.  We say that $f$ is
\emph{$\epsilon$-differentially private} if it is an $\epsilon$-sensitive
function from $\db$ to $DX$, where the set of distributions $DX$ is equipped
with the following distance, sometimes known as the \emph{max divergence}:
\begin{align}
  \label{eq:privacy-distance}
  \MD_X(\mu_1, \mu_2) & = \sum_{x \in X} \ln\left|\frac{\mu_1(x)}{\mu_2(x)}\right|.
\end{align}
(Here, we stipulate that $\ln|0/0| = 0$ and $\ln|p/0| = \ln|0/p| = \infty$ for
$p \neq 0$.)

To understand this definition, suppose that $D_1$ and $D_2$ are two databases at
distance 1---for instance, because they differ with respect to the data of a
single individual.  If $f$ is $\epsilon$-differentially private, the above
definition implies that $f(D_1)$ and $f(D_2)$ are at most $\epsilon$ apart. When
$\epsilon$ is large, the probabilities of each outcome in the result
distributions can vary widely.  This means that, by simply observing one output
of $f$, we might be able to guess with good confidence which of the databases
$D_1$ or $D_2$ was used to produce that output.  Conversely, if $\epsilon$ is
small, it is hard to tell which database was used because the output
probabilities will be close.  For this reason, it is common to view $\epsilon$
as a \emph{privacy loss}---the larger it is, the more privacy we are giving up
to reveal the output of $f$.

Besides providing a strong privacy guarantee, this formulation of closeness for
distributions provides two important properties.  First, we can \emph{compose}
differentially private algorithms without ruining their privacy guarantee.  Note
that $DX$ forms a monad, where the return and bind operations are given as
follows:
\begin{align}
  \label{eq:kleisli-distrReturn}
  \eta(x)
  & = y \mapsto
    \begin{cases}
      1 & \text{if $x = y$} \\
      0 & \text{otherwise}
    \end{cases} \\
  \label{eq:kleisli-distr}
  f^\dagger(\mu) & = y \mapsto \sum_{x \in X}\mu(x) \cdot f(x)(y).
\end{align}
Intuitively, the return operation produces a deterministic distribution, whereas
bind samples an element $x$ from $\mu$ and computes $f(x)$.  When composing
differentially private algorithms, their privacy loss can be soundly added
together:
\begin{theorem}
  \label{thm:sensitivity-of-kleisli}
  Suppose that $f : \db \to DX$ is $\epsilon_1$-differentially private and that
  $g : \db \to X \to DY$ is such that the mapping $\delta \to g(\delta)(x)$ is
  $\epsilon_2$-differentially private for every $x$.  Then the composite
  $h : \db \to DY$ defined as
  \begin{align*}
    h(\delta) & = g(\delta)^\dagger(f(\delta))
  \end{align*}
  is $(\epsilon_1 + \epsilon_2)$-differentially private.
\end{theorem}

The other reason why the privacy metric is useful is that it supports many
building blocks for differential privacy.  Of particular interest is the
\emph{Laplace mechanism}, which blurs a numeric result with noise drawn from the
two-sided Laplace distribution.  If $x \in \R$, let $\mathcal{L}(x)$ be the
distribution with density\footnote{We use here a Laplace distribution with scale
  1.}  $y \mapsto \frac{1}{2}e^{-|x-y|}$.
\begin{theorem}
  The Laplace mechanism $\mathcal{L}$ is a non-expansive function of type
  $\R \to D\R$.\footnote{The definitions do not quite match up our setting,
    since $\mathcal{L}$ is a continuous, and not discrete distribution.  The
    result can be put on firm footing by working with a discretized version of
    the Laplace distribution~\cite{DworkR14}.}
\end{theorem}
Thus, to define an $\epsilon$-differentially private numeric query on a
database, it suffices to define an $\epsilon$-sensitive, deterministic numeric
query, and then blur its result with Laplace noise.  Differential privacy
follows from the composition principles for sensitivity. This reasoning is
justified by the fact that the Laplace mechanism adds noise proportional to the
sensitivity of the numeric query in $L^1$ distance.

\subsection{Sensitivity as a Resource}

Because differential privacy is a sensitivity property, techniques for analyzing
the sensitivity of programs can also be used to analyze their privacy
guarantees.  One particularly successful approach in this space is rooted in
type systems inspired by linear logic, as pioneered by Reed and Pierce in the
Fuzz programming language~\cite{DBLP:conf/uss/HaeberlenPN11,
  DBLP:conf/icfp/ReedP10}.  At its core, Fuzz is just a type system for tracking
sensitivity.  Typing judgments are similar to common functional programming
languages, but variable declarations are of the form $x_i :_{r_i} \tau_i$:
\begin{mathpar}
  x_1 :_{r_1} \tau_1, \ldots, x_n :_{r_n} \tau_n \vdash e : \sigma.
\end{mathpar}  The
annotations $r_i \in \Dist$ are \emph{sensitivity indices}, whose purpose is to
track the effect that changes to the program input can have on its output: if we
have two substitutions $\gamma$ and $\gamma'$ for the variables $x_i$, then the
\emph{metric preservation} property of the Fuzz type system guarantees that
\begin{align}
  \label{eq:metric-preservation}
  d(e[\gamma/\vec{x}], e[\gamma'/\vec{x}])
  & \leq \sum_i r_i \cdot d(\gamma(x_i), \gamma'(x_i)),
\end{align}
where the metrics $d$ are computed based on the type of each expression and
value.  This means that we can bound the distance on the results of the two runs
of $e$ by adding up the distances of the inputs scaled by their corresponding
sensitivities.  When this bound is finite, the definition of the metrics
guarantees that the two runs have the same termination behavior.  When
$r_i = \infty$, the above inequality provides no guarantees if the value of
$x_i$ varies.

Fuzz includes data types commonly found in functional programming languages,
such as numbers, products, tagged unions, recursive types and functions.  The
typing rules of the language explain how the sensitivities of each variable must
be updated to compute each operation.  The simplest typing rule says that, in
order to use a variable, its declared sensitivity must be greater than 1:
\begin{mathpar}
  \inferrule
  { r \geq 1 }
  { \Gamma, x :_r \tau, \Delta \vdash x : \tau }
\end{mathpar}
As a more interesting example, to construct a pair $(e_1,e_2)$, the following
rule says that we need to add the sensitivities of the corresponding contexts:
\begin{mathpar}
  \inferrule
  { \Gamma_1 \vdash e_1 : \tau_1 \\ \Gamma_2 \vdash e_2 : \tau_2 }
  { \Gamma_1 + \Gamma_2 \vdash (e_1,e_2) : \tau_1 \otimes \tau_2 }.
\end{mathpar}
This behavior is a result of the distance of the tensor type $\otimes$: the
distance between two pairs in $\tau_1 \otimes \tau_2$ is the result of adding
the distances between the first and second components; therefore, the
sensitivity of each variable for the entire expression is the sum of the
sensitivities for each component.  In this sense, sensitivities in Fuzz behave
like a resource that must be distributed across all variable uses in a
program. For the sake of analogy, we might compare this treatment to how
fractional permissions work in separation logic: the predicate $l \mapsto_q x$
indicates that we own a fraction $q \in [0,1]$ of a resource stating that $l$
points to $x$. If $q = q_1 + q_2$, we can split this predicate as
$l \mapsto_{q_1} x * l \mapsto_{q_2} x$, allowing us to distribute this resource
between different threads.

The distance on $\otimes$ corresponds to the sum in the upper bound in the
statement of metric preservation (\Cref{eq:metric-preservation}).  This distance
is useful because it is the one that yields good composition principles for
differential privacy.  This can be seen in the typing rule for sampling from a
probabilistic distribution:
\begin{mathpar}
  \inferrule
  { \Gamma \vdash e_1 : \bigcirc \tau \\
    \Delta, x :_r \tau \vdash e_2 : \bigcirc \sigma }
  { \Gamma + \Delta \vdash \mLet{x}{e_1}{e_2} : \bigcirc \sigma }
\end{mathpar}
Here, $\bigcirc \tau$ denotes the type of probability distributions over values
of type $\tau$.  This operation samples a value $x$ from the distribution $e_1$
and uses this value to compute the distribution $e_2$.  We can justify the
soundness of this rule by reducing it to \Cref{thm:sensitivity-of-kleisli}: the
addition on contexts corresponds to the fact that the privacy loss of a program
degrades linearly under composition.

Besides the tensor product $\otimes$, Fuzz also features a \emph{with product}
$\with$, where the distances between components are combined by taking their
maximum.  This leads to a different typing rule for $\with$ pairs, which does
not add up the sensitivities:
\begin{mathpar}
  \inferrule
  { \Gamma \vdash e_1 : \tau_1 \\ \Gamma \vdash e_2 : \tau_2 }
  { \Gamma \vdash (e_1, e_2) : \tau_1 \with \tau_2 }
\end{mathpar}
If we compare these rules for pairs, we see a clear analogy with linear logic:
$\otimes$ requires us to combine contexts, whereas $\with$ allows us to share
them.
Fuzz's elimination rules for products continue to borrow from linear logic:
deconstructing a tensor gives both elements but deconstructing a with product
returns only one.
\begin{mathpar}
  \inferrule*[]{
    \Gamma \vdash e : \tau_1 \otimes \tau_2
    \and
    \Delta, x :_r \tau_1, y :_r \tau_2 \vdash e' : \tau'
  }{
    \Delta + r\Gamma \vdash \LetPairIn{x}{y}{e}{e'} : \tau'
  }
  \and
  \inferrule*[]{\Gamma \vdash e : \tau_1 \with \tau_2}{\Gamma \vdash \pi_i\;e : \tau_i}
\end{mathpar}
This partly explains why the connectives' distances involve addition and
maximum. When using a tensor product, both elements can affect how much the
output can vary, so both elements must be considered. (Note that Fuzz is an
affine type system: we are free to ignore one of the product's components, and
thus we can write projection functions out of a tensor product.)  When
projecting out of a with product, only one of the elements will affect the
program's output, so we only need to consider the component that yields the
maximum distance.

Fuzz uses the $!_s$ type for managing sensitivities. Intuitively, $!_s \tau$
behaves like $\tau$, but with the distances scaled by $s$; when $s = \infty$,
this means that different points are infinitely apart.  The introduction rule
scales the sensitivities of variables in the environment. This can be used in
conjunction with the elimination rule to propagate the sensitivity out of the
type and into the environment.
\begin{mathpar}
  \inferrule*[]{\Gamma \vdash e : \tau}{s\Gamma \vdash \!e : \!_s \tau}
  \and
  \inferrule*[]{
    \Gamma \vdash e : \!_s\tau
    \and
    \Delta, x :_{rs} \tau \vdash e' : \tau'
  }{
    \Delta + r\Gamma \vdash \LetBangIn{x}{e}{e'} : \tau'
  }
\end{mathpar}

Finally, the rules for the linear implication $\multimap$ are similar to the
ones from linear logic, but adjusted to account for sensitivities.
\begin{mathpar}
  \inferrule*[]
  {\Gamma \, x:_1 \tau \vdash e: \sigma}
  {\Gamma \vdash \lambda x. e : \tau \multimap \sigma }
  \and
  \inferrule*[]
  {\Gamma \vdash e : \tau\multimap\sigma \and \Delta \vdash e': \tau}
  {\Gamma + \Delta \vdash e\;e':\sigma}
\end{mathpar}
To introduce the linear implication $\multimap$, the bound variable needs to
have sensitivity 1. When eliminating $\multimap$, the environments need to be
added. In categorical language, addition, which is also present in the metric
for $\otimes$, is connected to the fact that there is an adjunction between the
functors $X \otimes (-)$ and $X \multimap (-)$.





\subsection{\lp distances}
\label{sec:lp-distances}

The $L^1$ and $L^\infty$ distances are instances of a more general family of
$L^p$ distances (for $p\in \R_\infty^{\geq 1}$).\footnote{The $L^p$ distances
  can be defined with $p\geq 0$ but for simplicity of our treatment we will only
  consider $p\geq1 $.}  Given a sequence of distances
$\vec{x} = (x_1,\dots,x_n) \in (\Dist)^n$, we first define the $L^p$
\emph{pseudonorm}\footnote{``pseudo-'' because it can be infinite.} as follows:
\begin{align*}
  ||\vec{x}||_p & = (\Sigma_{i=1}^{n} x_i^p)^{1/p}.
\end{align*}
This definition makes sense
whenever the distances $x_i$ and $p$ are finite.  When $p = \infty$, we define
the right-hand side as the limit $\max_{i=1}^n x_i$.  When $x_i = \infty$ for
some $i$, we define the right-hand side as $\infty$.
We have the following classical properties:
\begin{proposition}[H{\"o}lder inequality]\label{prop:Holder}
  For all $p, q\geq 1$ such that $\frac{1}{p}+\frac{1}{q}=1$, and for all
  $\vec{x}$, $\vec{y} \in (\Dist)^n$, we have:
  $ \Sigma_{ i=1}^{n} x_i y_i\leq ||\vec{x}||_p   ||\vec{y}||_q $.\\
  For $p=2$, $q=2$, this is the Cauchy-Schwarz inequality:
  $ \Sigma_{i=1}^{n} x_i y_i\leq ||\vec{x}||_2   ||\vec{y}||_2 $.
\end{proposition}
\begin{proposition}\label{prop:relations_Lp_norms}
  For $1 \leq p \leq q$ we have, for $\vec{x} \in (\Dist)^n$:
\begin{align}
  || \vec{x}||_q & \leq || \vec{x}||_p  \label{prop:relations_Lp_norms1}\\
  || \vec{x}||_p & \leq n^{\frac{1}{p}-\frac{1}{q}}|| \vec{x}||_q  \label{prop:relations_Lp_norms2}\\
  || \vec{x}||_2 & \leq || \vec{x}||_1 \leq \sqrt{n} \; ||\vec{x}||_2 \label{prop:relations_Lp_norms3}
\end{align}
\end{proposition}
The $L^p$ pseudonorms yield distances on tuples.  More precisely, suppose that
$(X_i)_{1\leq i \leq n}$ are metric spaces.  The following defines a metric on
$X=X_1 \times \dots \times X_n$:
$$d_p( \vec{x}, \vec{x}')=||(d _{X_1} (x_1,x'_1),\dots,d _{X_n} (x_n,x'_n))||_p  $$
\begin{proposition}\label{prop:relations_dp}
For $1 \leq p \leq q$ we have, for $\vec{x}, \vec{x'} \in X_1 \times \dots \times X_n$:
\begin{align}
d_q(\vec{x}, \vec{x'} )  & \leq d_p(\vec{x}, \vec{x'} ) \leq
n^{\frac{1}{p}-\frac{1}{q}}  d_q(\vec{x}, \vec{x'} )  \label{prop:relations_dp1}\\
d_2(\vec{x}, \vec{x'} )  & \leq d_1(\vec{x}, \vec{x'} ) \leq
\sqrt{n} \;  d_2(\vec{x}, \vec{x'} ) \label{prop:relations_dp2}
\end{align}
\end{proposition}

\section{\system{}: Programming with \lp Distances}
As we discussed earlier, the $L^1$ distance is not the only distance on products
with useful applications.  In the context of differential privacy, for example,
the $L^2$ distance is used to measure the sensitivity of queries when employing
the Gaussian mechanism, a method for private data release that sanitizes data by
adding Gaussian noise instead of Laplacian noise.\footnote{%
  Technically, the Gaussian mechanism is used to achieve a relaxation of
  differential privacy known as \emph{approximate}, or
  $(\epsilon,\delta)$-differential privacy.  Though this notion cannot be
  analyzed directly by classical verification techniques for differential
  privacy, it can be handled by recent extensions of
  Fuzz~\cite{DBLP:conf/lics/AmorimGHK19,DBLP:journals/pacmpl/NearDASGWSZSSS19}.}

It is possible to extend a Fuzz-like analysis with $L^2$ distances by adding
primitive types and combinators for vectors.  This was done, for instance, in
the Duet language~\cite{DBLP:journals/pacmpl/NearDASGWSZSSS19}, which provides
the Gaussian mechanism as one of the primitives for differential privacy.  Such
an extension can help verify a wide class of algorithms that manipulate vectors
in a homogeneous fashion, but it makes it awkward to express programs that
require finer grained access to vectors.

To illustrate this point, suppose that we have a non-expansive function
$f : \R^2 \to \R$, where the domain carries the $L^2$ metric.  Consider the
mapping
\[ g(x, y) = f(2x,y) + f(2y,x). \] How would we analyze the sensitivity of $g$?
We cannot translate such a program directly into a system like Duet, since it
does not allow us to manipulate $L^2$ vectors at the level of individual
components. However, we could rewrite the definition of $g$ to use matrix
operations, which could be easily incorporated in a variant of Duet.
Specifically, consider the following definition:
\[ g(\vec{x}) =
  f\left(
    \begin{bmatrix}
      2 & 0 \\
      0 & 1
    \end{bmatrix} \vec{x}
    \right)
  +
  f\left(
    \begin{bmatrix}
      0 & 2 \\
      1 & 0
    \end{bmatrix}
    \vec{x} \right).\] The $L^2$ sensitivity of a linear transformation
$\vec{x} \mapsto M \vec{x}$ can be easily computed if we know the coefficients
of the matrix $M$. Note that
\begin{align*}
  d(M\vec{x},M\vec y)
  & = ||M\vec x - M\vec y||_2
  = ||M(\vec x - \vec y)||_2
  = \frac{||M(\vec x - \vec y)||_2}{||\vec x - \vec y||_2}||\vec x - \vec
    y||_2 \\
  & \leq
    \left(\sup_{\vec z} \frac{||M\vec z||_2}{||\vec z||_2}\right)d(\vec x,\vec y).
\end{align*}
The quantity $\sup_{\vec{z}} ||M\vec{z}||_2 / ||\vec{z}||_2$, known as the
\emph{operator norm} of $M$, gives the precise sensitivity of the above
operation, and can be computed by standard algorithms from linear algebra.  In
the case of $g$, both matrices have a norm of $2$.  This means that we can
analyze the sensitivity of $g$ compositionally, as in Fuzz: addition is
$1$-sensitive in each variable, so we just have to sum the sensitivities of
$\vec{x}$ in each argument, yielding a combined sensitivity of $4$.
Unfortunately, this method of combining the sensitivities of each argument is
too coarse when reasoning with $L^p$ distances, which leads to an imprecise
analysis. To obtain a better bound, we can reason informally as follows.  First,
take
\[ M =
  \begin{bmatrix}
    2 & 0 \\
    0 & 1 \\
    0 & 2 \\
    1 & 0
  \end{bmatrix}. \]
  We can compute the operator norm of $M$ directly:
\begin{align*}
  ||M|| = \sup_{x,y}\frac{\sqrt{2^2x^2 + y^2 + 2^2y^2 + x^2}}{\sqrt{x^2+y^2}}
  = \sup_{x,y} \frac{\sqrt{5(x^2+y^2)}}{\sqrt{x^2+y^2}}
  = \sqrt 5,
\end{align*}
which implies that $M$ is a $\sqrt 5$-sensitive function of type
$\R^2 \to \R^4 \cong \R^2 \times \R^2$.  Moreover, thanks to
\Cref{prop:relations_dp}, we can view addition $(+)$ as a $\sqrt 2$-sensitive
operator of type $\R^2 \to \R$, since
\[ d_{\R}(x_1+x_2,y_1+y_2) \leq d_{\R}(x_1-y_1) + d_{\R}(x_2-y_2) = d_1(\vec
  x,\vec y) \leq \sqrt 2 d_2(\vec x,\vec y). \]
Thus, by rewriting the definition of $g$ as
\[ (+) \circ (f \times f) \circ M, \] where
$f \times f : \R^4 \cong \R^2 \times \R^2 \to \R \times \R$ denotes the
application of $f$ in parallel, we can compute the sensitivity of $g$ by
multiplying the sensitivity of each stage, as
$\sqrt 2 \times 1 \times \sqrt 5 = \sqrt{10} \approx 3.16$, which is strictly
better than the previous bound.

Naturally, we could further extend Fuzz or Duet with primitives for
internalizing this reasoning, but it would be preferable to use the original
definition of $g$ and automate the low-level reasoning about distances.  In this
section, we demonstrate how this can be done via \system{}, a language that
refines Fuzz by incorporating more general distances in its typing
environments. Rather assuming that input distances are always combined by
addition, or the $L^1$ distance, \system{} allows them to be combined with
arbitrary $L^p$ distances.  This refinement allows us to analyze different
components of a vector as individual variables, but also to split the
sensitivity of these variables while accounting for their corresponding vector
distances.  In the remaining of this section, we present the syntax and type
system of \system{}, highlighting the main differences with respect to the
original Fuzz design.  Later, in \Cref{sec:semantics}, we will give a semantics
to this language in terms of metric spaces, following prior
work~\cite{DBLP:conf/lics/AmorimGHK19}.

\paragraph{Types and Terms}

\Cref{fig:types-and-terms} presents the grammar of types and the main term
formers of \system{}. They are similar to their Fuzz counterparts; in
particular, there are types for real numbers, products, sums, functions, and a
unit type.  The main novelty is in the product type $\tau \otimes_p \sigma$,
which combines the metrics of each component using the $L^p$ distance
(cf. \Cref{sec:lp-distances}).  The types $\tau \otimes_1 \sigma$ and
$\tau \otimes_\infty \sigma$ subsume the types $\tau \otimes \sigma$ and
$\tau \with \sigma$ in the original Fuzz language.  Note that there is no term
constructor or destructor for the Fuzz type $\with$, since it is subsumed by
$\otimes_\infty$. The type $\tau \multimap_p \sigma$ represents non-expansive
functions endowed with a metric that is compatible with the $L^p$ metric, in
that currying works (cf. \Cref{sect:examples}). We will sometimes write
$\otimes$ for $\otimes_1$ and $\multimap$ for $\multimap_1$.

Another novelty with respect to Fuzz is that there are two constructors for
probability distributions, $\bigcirc_P$ and $\bigcirc_H$.  The first one carries
the original Fuzz privacy metric, while the second one carries the Hellinger
distance. As we will see shortly, the composition principle for the Hellinger
distance uses a contraction operator for the $L^2$ distance, which was not
available in the original Fuzz design.  Both distribution types feature term
constructors $\mathbf{mlet}$ and $\mathbf{return}$ for sampling from a
distribution and for injecting values into distributions.  To simplify the
notation, we do not use separate versions of these term formers for each type.

\begin{figure}[t]
  \centering
  \begin{align*}
    \tau, \sigma, \rho & ::= 1
      \mid \R
      \mid \!_s\tau
      \mid \bigcirc_P \tau
      \mid \bigcirc_H \tau
      \mid \tau \multimap_p \sigma
      \mid \tau \otimes_p \sigma
      \mid \tau \oplus \sigma
    & (p \in \R^{\ge1}_{\infty} , s \in \Sens) \\
    e & ::= \:x \mid r \in \R \mid () \mid \lambda x . e \mid e \: e
            \mid (e, e) \mid \LetPairIn x y {e} {e} \\
    & \mid \inj_i e \mid \left(\CaseOf{e}{x}{e}{y}{e}\right)
      \mid \!e \mid \LetBangIn{x}{e}{e} \\
    &\mid\mLet {x}{e}{e} \mid\return e \mid \cdots
  \end{align*}
  \caption{Types and terms in \system{}}
  \label{fig:types-and-terms}
\end{figure}

\paragraph{Bunches} Before describing its type system, we need to talk about how
typing environments are handled in \system{}. In the spirit of bunched logics,
environments are bunches defined with the following grammar:
$$\Gamma, \Delta ::= \cdot \mid [ x : \tau ]_s \mid \Gamma \,_p \Delta$$
The empty environment is denoted as $\cdot$. The form $[x : \tau]_s$ states that
the variable $x$ has type $\tau$ and sensitivity $s$. The form
$\Gamma\,_p\Delta$ denotes the concatenation of $\Gamma$ and $\Delta$, which is
only defined when the two bind disjoint sets of variables.  As we will see in
\Cref{sec:semantics}, bunches will be interpreted as metric spaces, and the $p$
index denote which \lp metric we will use to combine the metrics of $\Gamma$ and
$\Delta$.

The type system features several operations and relations on bunches, which are
summarized in \Cref{fig:bunchoperations}. We write
$\Gamma \leftrightsquigarrow \Gamma'$ to indicate that we can obtain $\Gamma'$
by rearranging commas up to associativity and commutativity, and by treating the
empty environment as an identity element; \Cref{fig:bunchoperations} has a
precise definition. Observe that associativity only holds for equal values of
$p$.  This operation will be used to state a permutation rule for the type
system of \system{}.

Like in Fuzz, environments have a scaling operation $s\Gamma$ which scales all
sensitivities in the bunch by $s$. For example,
$$
s ([x : \tau]_{r_1},_p[y : \sigma]_{r_2}) = ([x : \tau]_{s\cdot r_1},_p[y :
\sigma]_{s\cdot r_2}).
$$
The exact definition of scaling in such graded languages is subtle, since minor
variations can quickly lead to unsoundness.  The definition we are using
($\infty \cdot 0 = 0 \cdot \infty = \infty$), which goes back to prior
work~\cite{DBLP:conf/lics/AmorimGHK19}, is sound, but imprecise, since it leads
to too many variables being marked as $\infty$-sensitive. It would also be
possible to have a more precise variant that uses a non-commutative definition
of multiplication on distances~\cite{DBLP:conf/popl/AmorimGHKC17}, but we keep
the current formulation for simplicity. (For a more thorough discussion on these
choices and their tradeoffs, see \ifappendix \Cref{zero_and_infty}.) \else the
  "Zero and Infinity" example in Appendix B of the full version
  \cite{DBLP:journals/corr/abs-2202-01901} of this paper.) \fi

In the original Fuzz type system, rules with several premises usually have their
environments combined by adding sensitivities pointwise, which corresponds to a
use of the $L^1$ metric.  In \system{}, we have instead a family of contraction
operations $Contr(p,\Gamma,\Delta)$ for combining environments, one for each \lp
metric.  Contraction only makes sense if $\Gamma$ and $\Delta$ differ only in
sensitivities and variable names, but have the same structure otherwise.  We
write this relation as $\Gamma \approx \Delta$.  When contracting two leaves,
sensitivities are combined using the \lp norm, while keeping variable names from
the left bunch.

Unlike Fuzz, where contraction is implicit in rules with multiple premises,
\system{} has a separate, explicit contraction typing rule. The rule will be
stated using the $vars$ function, which lists all variables in a bunch.

\begin{figure}
  \begin{minipage}{.49\textwidth}
    \begin{align*}
      vars(\cdot) &= []\\
      vars([x : \tau]_s) &= [x]\\
      vars((\Gamma_1 ,_p \Gamma_2)) &= vars(\Gamma_1) \mdoubleplus vars(\Gamma_2)
    \end{align*}
  \end{minipage}
  \begin{minipage}{.49\textwidth}
    \begin{align*}
      \cdot &\approx \cdot \\
      [x : \tau]_s &\approx [y : \sigma]_r& \text{if } \tau = \sigma\\
      \Gamma_1 \,_p \Gamma_2 &\approx \Delta_1 ,_q \Delta_2
        & \text{if } p = q \wedge \Gamma_i \approx \Delta_i
    \end{align*}
  \end{minipage}

  \begin{minipage}{.55\textwidth}
    \begin{align*}
      \Gamma &\leftrightsquigarrow \Delta
        & \text{if } \Gamma = \Delta\\
      \Gamma &\leftrightsquigarrow \cdot ,_p \Delta
        & \text{if } \Gamma \leftrightsquigarrow \Delta\\
      \Gamma &\leftrightsquigarrow \Delta ,_p \cdot
        & \text{if } \Gamma \leftrightsquigarrow \Delta\\
      \Gamma_1 ,_p \Gamma_2 & \leftrightsquigarrow \Delta_1 ,_p \Delta_2
        & \text{if } \Gamma_i \leftrightsquigarrow \Delta_i \\
      \Gamma_1 ,_p \Gamma_2 &\leftrightsquigarrow \Delta_2 ,_p \Delta_1
        & \text{if } \Gamma_i \leftrightsquigarrow \Delta_i\\
      \Gamma_1 ,_p (\Gamma_2 ,_p \Gamma_3) &\leftrightsquigarrow (\Delta_1 ,_p \Delta_2) ,_p \Delta_3
        & \text{if } \Gamma_i \leftrightsquigarrow \Delta_i \\
      \Gamma_2 &\leftrightsquigarrow \Gamma_1
        & \text{if } \Gamma_1 \leftrightsquigarrow \Gamma_2
    \end{align*}
  \end{minipage}
  \begin{minipage}{.45\textwidth}
    \begin{align*}
      s \; \cdot &= \cdot \\
      s \: [\tau]_r &= [\tau]_{s \cdot r} \\
      s \: (\Gamma \,_p \Delta) &= s\Gamma \,_p s\Delta
    \end{align*}
  \end{minipage}

  \begin{align*}
    c(p,q)
    & =
      \begin{cases}
        1 & \text{if $p = \infty$} \\
        2^{\left|\frac1q - \frac1p\right|} & \text{otherwise}
      \end{cases} \\
    Contr(p, \cdot , \cdot)
    &= \cdot\\
    Contr(p, [x: \tau]_s , [y:\tau]_r)
    &= [x:\tau]_{ \sqrt[p]{s^p + r^p} }\\
    Contr(p, (\Gamma_1 ,_q \Gamma_2) , (\Delta_1 ,_q \Delta_2))
    &=
      c(p,q)(Contr(p, \Gamma_1, \Delta_1) ,_q Contr(p, \Gamma_2, \Delta_2)).
  \end{align*}
  \caption{Bunch Operations}\label{fig:bunchoperations}
\end{figure}




\paragraph{Type System} Our type system is similar to the one of Fuzz, but
adapted to use bunched environments. The typing rules are displayed on
\Cref{fig:BunchedFuzzTypingRules}.
For example, in the $\otimes$I rule, notice
that the $p$ on the tensor type is carried over to the bunch in the resulting
environment.
 Similarly, in the $\multimap$I rule,
the value of $p$ that annotates the bunch in the premise is carried over to the
$\multimap_p$ in the conclusion.

Like in Fuzz, the $!$E rule propagates the scaling factor, but using the bunch
structure. Rather than adding the two environments, we splice one into the
other: the notation $\Gamma(\Delta)$ denotes a compound bunch where we plug in
the bunch $\Delta$ into another bunch $\Gamma(\star)$ that has a single,
distinguished hole $\star$. As we mentioned earlier, \system{} has an explicit
typing rule for contraction, whereas contraction in Fuzz is implicit in rules
with multiple premises. Note also that we have unrestricted weakening.
%
%
Finally, we have the rules for typing the return and bind primitives of the
probabilistic types $\bigcirc_H$ and $\bigcirc_P$.  Those for $\bigcirc_P$ are
adapted from Fuzz, by using contraction instead of adding up the environments.
The ones for $\bigcirc_H$ are similar, but use $L^2$ contraction instead, since
that is the metric that enables composition for the Hellinger distance.

\begin{figure}
\begin{mathpar}
  \inferrule*[Right=Axiom]
    {s \ge 1}
  {[x : \tau]_s \vdash x : \tau}
  \and
  \inferrule*[Right=\RI]
    { }
    {\cdot \vdash r : \R}
  \and
  \inferrule*[Right=1I]{ }{\cdot \vdash () : 1}
  \\
  \inferrule*[Right=\lolliI]
    {\Gamma \,_p [x : \tau]_1 \vdash e : \sigma}
    {\Gamma \vdash \lambda x. e : \tau \multimap_p \sigma}
  \and
  \inferrule*[Right=\lolliE]
    {\Gamma \vdash f : \tau \multimap_p \sigma \and \Delta \vdash e : \tau}
    {\Gamma ,_p \Delta \vdash f \: e : \sigma}
  \\
  \inferrule*[Right=\tensorI]
    {\Gamma \vdash e_1 : \tau \and \Delta \vdash e_2 : \sigma}
    {\Gamma \,_p \Delta \vdash (e_1, e_2) : \tau \otimes_p \sigma}
  \and
  \inferrule*[Right=\tensorE]
    { \Delta \vdash e_1 : \tau \otimes_p \sigma
      \and
      \Gamma([x : \tau]_s \,_p [y : \sigma]_s) \vdash e_2 : \rho}
    {\Gamma(s\Delta) \vdash \LetPairIn{x}{y}{e_1}{e_2} : \rho}
  \\
  \inferrule*[Right=\plusI1]
    {\Gamma \vdash e : \tau}
    {\Gamma \vdash \inj_1 e : \tau \oplus \sigma}
  \and
  \inferrule*[Right=\plusI2]
    {\Gamma \vdash e : \sigma}
    {\Gamma \vdash \inj_2 e : \tau \oplus \sigma}
  \and
  \inferrule*[Right=\plusE]
    {
      \Gamma \vdash e_1 : \tau \oplus \sigma
      \and \Delta([x : \tau]_s) \vdash e_2 : \rho
      \and \Delta([y : \sigma]_s) \vdash e_3 : \rho }
    {\Delta(s\Gamma) \vdash \CaseOf {e_1}{x}{e_2}{y}{e_3} : \rho}
  \\
  \inferrule*[Right=!I]
    {\Gamma \vdash e : \tau}
    {s\Gamma \vdash \! e : \!_s \tau}
  \and
  \inferrule*[Right=!E]
    {\Gamma \vdash e_1 : \!_r \tau \and \Delta([x : \tau]_{rs}) \vdash e_2 : \sigma}
    {\Delta(s\Gamma) \vdash \LetBangIn{x}{e_1}{e_2} : \sigma}
  \and
  \inferrule*[Right=Contr]
    {\Gamma(\Delta \,_p \Delta') \vdash e : \tau \and \Delta \approx \Delta'}
    {\Gamma(Contr(p, \Delta, \Delta')) \vdash e[vars(\Delta')/vars(\Delta)] : \tau}
  \and
  \inferrule*[Right=Weak]
    {\Gamma(\cdot) \vdash e : \tau}
    {\Gamma(\Delta) \vdash e : \tau}
  \and
  \inferrule*[Right=Exch]
    {\Gamma \vdash e : \tau \and \Gamma \leftrightsquigarrow \Gamma'}
    {\Gamma' \vdash e : \tau}
  \\
  \inferrule*[Right=Bind-P]
    { \Gamma \approx \Delta
      \\\\
      \Gamma \vdash e_1 : \bigcirc_P \tau
      \and
      \Delta ,_p [x : \tau]_s \vdash e_2 : \bigcirc_P \sigma
      }
    {Contr(1, \Gamma, \Delta) \vdash \mLet {x}{e_1}{e_2} : \bigcirc_P \sigma}
  \and
  \inferrule*[Right=Return-P]
    { \Gamma \vdash e : \tau}
    { \infty\Gamma \vdash \return e : \bigcirc_P \tau}
  \and
  \inferrule*[Right=Bind-H]
    { \Gamma \approx \Delta
      \\\\
      \Gamma \vdash e_1 : \bigcirc_H \tau
      \and
      \Delta ,_p [x : \tau]_s \vdash e_2 : \bigcirc_H \sigma
      }
    {Contr(2, \Gamma, \Delta) \vdash \mLet {x}{e_1}{e_2} : \bigcirc_H \sigma}
  \and
  \inferrule*[Right=Return-H]
    { \Gamma \vdash e : \tau}
    { \infty\Gamma \vdash \return e : \bigcirc_H \tau}
\end{mathpar}
\caption{Bunched Fuzz typing rules}\label{fig:BunchedFuzzTypingRules}
\end{figure}

Let us now explain in which sense $\otimes_\infty$ corresponds to the $\with$ connective of Fuzz. We will need the following lemma:
\begin{lemma}[Renaming]\label{lem:renaming}
Assume that there is a type derivation of $\Gamma \vdash e : \tau $ and that $\Gamma \approx \Gamma'$. Then there exists a derivation of
 $\Gamma' \vdash e[vars(\Gamma')/vars(\Gamma)]  : \tau $.
\end{lemma}
Now, the $\with$ connective in Fuzz supports two operations, projections and pairing. The connective  $\otimes_\infty$  of \system\; also supports these operations, but as derived forms.  First, projections can be encoded by defining $\pi_i(e)$ for $i=1,2$  as  $\LetPairIn{x_1}{x_2}{e}{x_i}$.
Second, for pairing assume we have two derivations of
$\Gamma \vdash e_i:\sigma_i$ for $i=1,2$, and let $\Gamma'$ be an environment
obtained from $\Gamma$ by renaming all variables to fresh ones. Then we have
$\Gamma \approx \Gamma'$ and thus
$$
\inferrule*[Right=Contr]
{\inferrule*[Right=\tensorI]
    {\Gamma \vdash e_1 : \sigma_1 \and
      \inferrule*[Right=\Cref{lem:renaming}]
      {\Gamma \vdash e_2 : \sigma_2 \and \Gamma \approx \Gamma'}
      {\Gamma' \vdash e_2[vars(\Gamma')/vars(\Gamma)]  : \sigma_2 }}
    {\Gamma \,_{\infty}  \Gamma' \vdash (e_1, e_2[vars(\Gamma')/vars(\Gamma)] ) : \sigma_1 \otimes_{\infty} \sigma_2}
    }
      {Contr(\infty,\Gamma,\Gamma') \vdash (e_1, e_2 ) : \sigma_1 \otimes_{\infty} \sigma_2}
$$
Note that we have defined $\sqrt[\infty]{x^\infty + y^\infty} = \max(x,y)$ by
taking the limit of $\sqrt[p]{x^p + y^p}$ when $p$ goes to infinity, and thus we
have $Contr(\infty,\Gamma,\Gamma') = \Gamma$. Therefore the pairing rule of
$\with$ is derivable for $\otimes_\infty$.

\section{Semantics}
\label{sec:semantics}

Having defined the syntax of \system{} and its type system, we are ready to
present its semantics. We opt for a denotational formulation, where types $\tau$
and bunches $\Gamma$ are interpreted as metric spaces $\intrp\tau$ and
$\intrp\Gamma$, and a derivation $\pi$ of $\Gamma \vdash e : \tau$ is
interpreted as a non-expansive function
$\intrp\pi : \intrp\Gamma \to \intrp\tau$.  For space reasons, we do not provide
an operational semantics for the language, but we foresee no major difficulties
in doing so, since the term language is mostly inherited from Fuzz, which does
have a denotational semantics proved sound with respect to an operational
semantics~\cite{DBLP:conf/popl/AmorimGHKC17}.

\paragraph{Types}
Each type $\tau$ is interpreted as a metric space $\intrp\tau$ in a
compositional fashion, by mapping each type constructor to the corresponding
operation on metric spaces defined in \Cref{fig:type-interp}.  We now explain
these definitions.

The operations of the first four lines of \Cref{fig:type-interp} come from prior
work on Fuzz~\cite{DBLP:conf/popl/AmorimGHKC17, DBLP:conf/lics/AmorimGHK19}. The
definition of $\otimes_p$ uses as carrier set the cartesian product, just as
$\otimes$ in previous works, but endows it with the \lp distance, defined in
\Cref{sec:lp-distances}. In the particular case of $p=1$, $\otimes_1$ is the
same as $\otimes$.

As for $\multimap_p$, we want to define it in such a way that currying and
uncurrying work with respect to $\otimes_p$, which will allow us to justify the
introduction and elimination forms for that connective.  For that we first
choose as carrier set the set $A \multimap B$ of non-expansive functions from
$A$ to $B$.
This set carries the metric
\begin{align}
\begin{split}
  & d_{A \multimap_p B}(f, g)\\
  & = \inf \{ r\in \Dist \mid \forall x,y \in A,
    d_B(f(x), g(y)) \le \sqrt[p]{r^p + d_A(x,y)^p }\}
    \label{eq:multimap-dist}
\end{split}
\end{align}
This metric is dictated by the type of the application operator in the \lp{}
norm: $(A \multimap_p B) \otimes_p A \multimap B$.  Intuitively, if $f$ and $g$
are at distance $r$, and we want application to be non-expansive, we need to
satisfy
\[ d_B(f(x), g(y)) \le \sqrt[p]{r^p + d_A(x,y)^p} \] for every $x, y \in A$.
The above definition says that we pick the distance to be the smallest possible
$r$ that makes this work. Note that this choice is forced upon us: in
category-theoretic jargon, the operations of currying and uncurrying, which are
intimately tied to the application operator, correspond to an adjunction between
two functors, which implies that any other metric space that yields a similar
adjunction with respect to $\otimes_p$ must be isomorphic to $\multimap_p$.  In
particular, this implies that its metric will be the same as the one of
$\multimap_p$.


For $\bigcirc_P A$ and $\bigcirc_H A$ the carrier set is the set $D A$ of discrete distributions over $A$.
As to the metric on the carrier set, the interpretation of $\bigcirc_P$ uses the max divergence, used in
the definition of differential privacy (see Sect. \ref{sec:dp}).  The interpretation of $\bigcirc_H$ uses
instead the Hellinger distance (see e.g. \cite{DBLP:conf/lics/AmorimGHK19}):
\begin{align}
  \HD_A(\mu, \nu) & \triangleq \sqrt{
                    \frac 1 2 \sum_{x \in A}
                    | \sqrt{\mu(x)} - \sqrt{\nu(x)} | ^2
                    }
  \label{eq:hellinger-dist}
\end{align}
\begin{figure}
  \begin{center}
    \begin{tabular}{c | c | c}
      Space $X$& $|X|$ &$d_X(x, y)$ \\
      \hline
      $1$ & $\{ * \}$ & 0 \\
      $\R$ & $\R$ & $|x - y|$ \\
      & &
          \multirow{3}{*}{$\begin{cases}
          s\cdot d_A(x, y) \text{ if }s\neq \infty\\
          \infty \text{ if }s=\infty, x\neq y \in A\\
          0 \text{ if }s=\infty, x= y \in A\\
        \end{cases}$}\\
   $!_sA$& $|A|$&\\
  &&\\
        &  &\multirow{3}{*}{$\begin{cases}
          d_A(x, y)\text{ if }x, y \in A\\
          d_B(x, y)\text{ if }x, y \in B\\
          \text{else }\infty\\
        \end{cases}$}\\
  $A \oplus B$&$|A| + |B|$&\\
  &&\\
      $A \otimes_p B$ & $|A|\times|B|$ & $\sqrt[p]{ d_A(\pi_1(x), \pi_1(y))^p +
        d_B(\pi_2(x),\pi_2(y))^p } $\\
      $A \multimap_p B$ & $A \multimap B$ & cf. \Cref{eq:multimap-dist} \\
      $\bigcirc_P A$ & $DA$ & $\MD_A(x,y)$; cf. \Cref{eq:privacy-distance} \\
      $\bigcirc_H A$ & $DA$ & $\HD_A(x,y)$; cf. \Cref{eq:hellinger-dist}
    \end{tabular}
  \end{center}
  \caption{Operations on metric spaces for interpreting types}\label{fig:type-interp}
  \end{figure}
\paragraph{Bunches}

The interpretation of bunches is similar to that of types. Variables correspond
to scaled metric spaces, whereas $,_p$ corresponds to $\otimes_p$:
\begin{align*}
  \intrp \cdot & = 1 &
  \intrp {[ x : \tau ]_s } & = {!_s \intrp {\tau}} &
  \intrp {\Gamma_1 \,_p \Gamma_2 } & =\intrp {\Gamma_1}\otimes_p \intrp
                                     {\Gamma_2}.
\end{align*}

One complication compared to prior designs is the use of an explicit exchange
rule, which is required to handle the richer structure of contexts.
Semantically, each use of exchange induces an isomorphism of metric spaces:

\begin{restatable}{theorem}{bunchmonoidisomorphism}
  \label{thm:bunch-monoid-isomorphism}
  Each derivation of $\Gamma \leftrightsquigarrow \Delta$ corresponds to an
  isomorphism of metric spaces $\intrp \Gamma \cong \intrp \Delta$.
\end{restatable}

Before stating the interpretation of typing derivations, we give an overview of
important properties of the above constructions that will help us prove the
soundness of the interpretation.

\paragraph{Scaling}

Much like in prior work~\cite{DBLP:conf/popl/AmorimGHKC17,
  DBLP:conf/lics/AmorimGHK19}, we can check the following equations:
\begin{proposition}
  \label{prop:scaling-distr}
  \begin{align*}
    !_{s_1} !_{s_2} A & = {\!_{s_1\cdot s_2} A}
    & !_s (A \oplus B) & =  \!_s A \oplus  \!_s B
    & !_s (A \otimes_p B) & =  \!_s A \otimes_p  \!_s B.
  \end{align*}
\end{proposition}

Moreover, an $s$-sensitive function from $A$ to $B$ is the same thing as a
non-expansive function of type $!_s A \multimap B$.

\begin{proposition}
  \label{prop:scaled-bunch-comm}
  For every bunch $\Gamma$, we have $\intrp{s\Gamma} = {!_s \intrp\Gamma}$.
\end{proposition}


\paragraph{Tensors} The properties on \lp distances allow us to relate product
types with different values of $p$.

\begin{proposition}\label{prop:relationTensors} [Subtyping of tensors]
  \begin{enumerate}
  \item Let $A$, $B$ be two metric spaces and $p, q\in \R^{\ge1}_{\infty}$ with $p\leq q$.
    Then the identity map on pairs belongs to the two following spaces:
    \begin{align*}
      A \otimes_p B & \multimap A \otimes_qB
      & !_{2^ {1/p-1/q}} (A \otimes_q B) & \multimap A \otimes_pB.
    \end{align*}
  \item In particular, when $p=1$ and $q=2$, the identity map belongs to:
    \begin{align*}
      A \otimes_1 B & \multimap A \otimes_2B
      & !_{\sqrt {2}}(A \otimes_2 B) & \multimap A \otimes_1B.
    \end{align*}
  \end{enumerate}
\end{proposition}
\begin{proof}
  For (1), the fact that the identity belongs to the first space follows from
  the fact that $d_q (x,y) \leq d_p(x,y)$, by \Cref{prop:relations_dp}
  (\Cref{prop:relations_dp1}).  The second claim is derived from
  \Cref{prop:relations_dp} (\Cref{prop:relations_dp1}) in the case $n=2$.
\end{proof}
\begin{remark}\label{remark:tensorinclusions}
  \Cref{prop:relationTensors} allows us to relate different spaces of functions
  with multiple arguments. For example,
  \begin{align*}
    (A \otimes_2 B \multimap
    C) & \subseteq  (A \otimes_1 B \multimap C) \\
    (A \otimes_1 B \multimap C)
       & \subseteq (!_{\sqrt {2}} (A \otimes_2 B) \multimap C).
  \end{align*}
  \system{} does not currently exploit these inclusions in any significant way,
  but we could envision extending the system with a notion of subtyping to
  further simplify the use of multiple product metrics in a single program.
\end{remark}

We also have the following result, which is instrumental to prove the soundness
of the contraction rule.

\begin{proposition}
  \label{prop:tensor-distr}
  Let $X, Y, Z, W$ be metric spaces, and $p, q \in \R_\infty^{\geq 1}$ with
  $p \neq \infty$.  The canonical isomorphism of sets
  \[ (X \times Y) \times (Z \times W) \cong (X \times Z) \times (Y \times W),\]
  which swaps the second and third components, is a non-expansive function of
  type
  \[ !_{c(p,q)}((X \otimes_q Y) \otimes_p (Z \otimes_q W)) \to (X \otimes_p Z)
    \otimes_q (Y \otimes_p W), \] where $c(p,q)$ is defined as in
  \Cref{fig:bunchoperations}.
\end{proposition}

\begin{proof}
  First, suppose that $p \leq q$.  Then we can write the isomorphism as a
  composite of the following non-expansive functions:
  \begin{align*}
    & !_{c(p,q)}((X \otimes_q Y) \otimes_p (Z \otimes_q W) \\
    & \to {!_{c(p,q)}((X \otimes_q Y) \otimes_q (Z \otimes_q W))}
    & \text{\Cref{prop:relationTensors}} \\
    & \cong {!_{c(p,q)}((X \otimes_q Z) \otimes_q (Y \otimes_q W))}
    & \text{assoc., comm. of $\otimes_q$} \\
    & = {!_{c(p,q)}(X \otimes_q Z)} \otimes_q {!_{c(p,q)} (Y \otimes_q W)}
    & \text{\Cref{prop:scaling-distr}} \\
    & = (X \otimes_p Z) \otimes_q (Y \otimes_p W)
    & \text{\Cref{prop:relationTensors}}.
  \end{align*}
  Otherwise, $p > q$, and we reason as follows.
  \begin{align*}
    & !_{c(p,q)}((X \otimes_q Y) \otimes_p (Z \otimes_q W) \\
    & \to {!_{c(p,q)}((X \otimes_p Y) \otimes_q (Z \otimes_p W))}
    & \text{\Cref{prop:relationTensors}} \\
    & \cong {!_{c(p,q)}((X \otimes_p Z) \otimes_p (Y \otimes_p W))}
    & \text{assoc., comm. of $\otimes_p$} \\
    & = (X \otimes_p Z) \otimes_q (Y \otimes_p W)
    & \text{\Cref{prop:relationTensors}}.
  \end{align*}
\end{proof}
One can then prove the following property:
\begin{restatable}{proposition}{contrsound}
  \label{prop:contrsound}
  Suppose that we have two bunches $\Gamma \approx \Delta$.  The carrier sets of
  $\intrp\Gamma$ and $\intrp\Delta$ are the same. Moreover, for any $p$, the
  diagonal function $\delta(x) = (x,x)$ is a non-expansive function of type
  \[ \intrp{Contr(p,\Gamma,\Delta)} \to \intrp\Gamma \otimes_p \intrp\Delta. \]
\end{restatable}

\paragraph{Function Types}
The metric on $\multimap_p$ can be justified by the following result:
\begin{proposition}\label{prop:adjunction}
  For every metric space $X$ and every $p \in \R^{\geq 1}_\infty$, there is an
  adjunction of type $(-) \otimes_p X \dashv X \multimap_p (-)$ in $\Met$ given
  by currying and uncurrying.  (Both constructions on metric spaces are extended
  to endofunctors on $\Met$ in the obvious way.)
\end{proposition}

Because right adjoints are unique up to isomorphism, this definition is a direct
generalization of the metric on functions used in
Fuzz~\cite{DBLP:conf/icfp/ReedP10, DBLP:conf/popl/AmorimGHKC17,
  DBLP:conf/lics/AmorimGHK19}, which corresponds to $\multimap_1$.

\begin{restatable}{theorem}{supdistance}
  \label{thm:d-iff-dhat}
  Suppose that $A$ and $B$ are \emph{proper} metric spaces, and let
  $f, g : A \to B$ be non-expansive.  Then
  $d_{A \multimap_1 B}(f, g) = \sup_x d_B(f(x),g(x))$.
\end{restatable}

We conclude with another subtyping result involving function spaces.

\begin{restatable}{theorem}{lollisubtyping}
  \label{thm:lolli-1-le-lolli-p}
  For all non-expansive functions $f, g \in A \to B$ and $p\geq1$, we have
  $d_{A \multimap_1 B}(f,g) \le d_{A \multimap_p B}(f,g).$ In particular, the
  identity function is a non-expansive function of type
  $(A \multimap_p B) \to (A \multimap_1 B)$.
\end{restatable}

\paragraph{Probability Distributions}

Prior work~\cite{DBLP:conf/lics/AmorimGHK19} proves that the return and bind
operations on probability distributions can be seen as non-expansive functions:
\begin{align*}
  \eta & : {!_\infty A} \to \bigcirc_PA \\
  (-)^\dagger(-) & : (!_\infty A \multimap_1 \bigcirc_PB) \otimes_1 \bigcirc_PA
                   \to \bigcirc_P B.
\end{align*}
These properties ensure the soundness of the typing rules for $\bigcirc_P$ in
Fuzz, and also in \system{}.  For $\bigcirc_H$, we can use the following
composition principle.
\begin{restatable}{theorem}{bindsoundness}
  \label{thm:bind-soundness}
  The following types are sound for the monadic operations on distributions,
  seen as non-expansive operations, for any $p\geq 1$:
  \begin{align*}
    \eta
    & : {!_\infty A} \to \bigcirc_H A \\
    (-)^\dagger(-)
    & : (!_\infty A \multimap_p \bigcirc_H B) \otimes_2 \bigcirc_H A \to
      \bigcirc_H B.
  \end{align*}
\end{restatable}

\paragraph{Derivations} Finally, a derivation tree builds a function from the context's space to the subject's space.
In the following definition, we use the metavariables $\gamma$ and $\delta$ to
denote variable assignments---that is, mappings from the variables of
environments $\Gamma$ and $\Delta$ to elements of the corresponding metric
spaces.  We use $\gamma(\delta)$ to represent an assignment in
$\intrp{\Gamma(\Delta)}$ that is decomposed into two assignments $\gamma(\star)$
and $\delta$ corresponding to the $\Gamma(\star)$ and $\Delta$ portions.
Finally, we use the $\lambda$-calculus notation $f \; x$ to denote a function
$f$ being applied to the value $x$.

\begin{definition}
  \label{def:derivation-semantics}
  Given a derivation $\pi$ proving $\:\Gamma \vdash e : \tau$, its
  interpretation $\intrp{\pi}\in \intrp \Gamma \rightarrow \intrp \tau$ is given
  by structural induction on $\pi$ as follows:

\begin{tabular}{ll}
  $\intrp {Axiom} \triangleq \lambda x .\; x $ &$\intrp {\R I } \triangleq \lambda () .\; r \in \R $ \\
 $\intrp {\multimap I \;\pi} \triangleq \lambda \gamma .\; \lambda x .\; \intrp \pi \;(\gamma , x) $
 & $ \intrp{\multimap E \; \pi_1 \; \pi_2} \triangleq \lambda (\gamma, \delta).\; \intrp {\pi_2} \; \gamma \; ( \intrp{\pi_1} \; \delta )$\\
 $\intrp {1I} \triangleq \lambda () .\; () $
  & $\intrp {\otimes I \; \pi_1 \; \pi_2}  \triangleq \lambda (\gamma , \delta) .\; (\intrp {\pi_1} \; \gamma) , (\intrp {\pi_2} \; \delta)$\\
&$\intrp {\otimes E \; \pi_1 \; \pi_2}
    \triangleq \lambda \gamma(\delta) .\; \intrp {\pi_2} \; \gamma(\intrp {\pi_1} \delta )$ \\
$\intrp {\oplus_iI \; \pi } \triangleq \lambda \gamma .\; \inj_i \intrp {\pi} \;\gamma $&
$\intrp {\oplus E \; \pi_1 \; \pi_2 } \triangleq \lambda \delta (\gamma).\; [ \intrp {\pi_2}  ,\intrp {\pi_3}] (\delta(\intrp {\pi_1}\gamma))$\\
$\intrp {!I \; \pi} \triangleq \intrp{\pi}$&  $\intrp {!E \; \pi_1 \; \pi_2} \triangleq \lambda \; \delta(\gamma).\;
    \intrp{\pi_2} \; \delta(\intrp {\pi_1}\; \gamma)$\\
 $\intrp {Contr \; \pi } \triangleq \lambda \gamma(\delta).\; \intrp \pi \; \gamma( \delta, \delta )$ &
$\intrp {Weak \; \pi } \triangleq \lambda \gamma(\delta) .\; \intrp \pi \; \gamma( \: () \: )$\\
$\intrp {Exch \; \pi } \triangleq \lambda \gamma'. \intrp{\pi} \phi_{\gamma' \slash \gamma}(\gamma') $ &
$\intrp {\mbox{Bind-P} \; \pi_1 \; \pi_2 } \triangleq \lambda \gamma' .\; {(\intrp {\pi_2}\gamma')}^\dagger(\intrp {\pi_1}\gamma')$\\
$\intrp {\mbox{Return-P} \; \pi } \triangleq \lambda \gamma .\;
    \eta(\intrp{\pi}\; \gamma)$ &
\end{tabular}
%

where in $\intrp {Exch \; \pi }$, the map $\phi_{\Gamma' \slash \Gamma}$ is the
isomorphism defined by \Cref{thm:bunch-monoid-isomorphism}.  and for the two
last cases see definitions in equations (\ref{eq:kleisli-distrReturn}) and
(\ref{eq:kleisli-distr}) (Bind-H and Return-H are defined in the same way).
\end{definition}

\begin{restatable}[Soundness]{theorem}{soundness}
  \label{thm:derivation-semantics-non-expansive}
  Given a derivation $\pi$ proving $\Gamma \vdash e : \tau$, then $\intrp \pi$
  is a non-expansive function from the space $\intrp \Gamma$ to the space
  $\intrp \tau$.
\end{restatable}

\section{Examples}\label{sect:examples}

\input{sect_examplesESOP}

\ifalgrules
\section{Implementation}
\input{sect_implementation.tex}
\fi

\section{Related Work}

\system{} is inspired by BI, the logic of bunched
implications~\cite{DBLP:journals/bsl/OHearnP99}, which has two
connectives for combining contexts.  Categorically, one of these
connectives corresponds to a Cartesian product, whereas the other
corresponds to a monoidal, or tensor product.  While related to linear
logic, the presence of the two context connectives allows BI to derive
some properties that are not valid in linear logic.
For example,
the cartesian product does not distribute over sums in linear logic
but it does distribute over sums in BI.

We have shown how the rules for such type systems are reminiscent of the ones
used in type systems for the calcuclus of bunched
implications~\cite{DBLP:journals/jfp/OHearn03}, and for reasoning about
categorical grammars~\cite{DBLP:series/lncs/6850}.  Specifically, O'Hearn
introduces a type system with two products and two
arrows~\cite{DBLP:journals/jfp/OHearn03}. Typing environments are bunches of
variable assignments with two constructors, corresponding to the two
products. Our work can be seen as a generalization of O'Hearn's work to handle
multiple products and to reason about program sensitivity.

Moot and Retor\'e \cite[Chapter 5]{DBLP:series/lncs/6850} introduce the multimodal
Lambek calculus, which extends the non-associative Lambek calculus, a classical
tool for describing categorical grammars. This generalization uses an indexed
family of connectives and trees to represent environments.  The main differences
with our work are: our indexed products are associative and commutative, while
theirs are not; our type system is affine; our type system includes a monad for
probabilities which does not have a correspondent construction in their logic;
our type system also possesses the graded comonad $!_s$ corresponding to the $!$
modality of linear logic, the interaction between this comonad and the bunches
is non-trivial and it requires us to explicitly define a notion of
contraction. Besides the fact that the main properties we study, metric
interpretation and program sensitivity, are very different from the ones studied
by the above authors, there are some striking similarities between the two
systems.

A recent work by Bao et al.~\cite{10.1145/3498719} introduced a novel bunched
logic with indexed products and magic wands with a preorder between the
indices. This logic is used as the assertion logic of a separation logic
introduced to reason about negative dependence between random variables. The
connectives studied in this work share some similarities with the ones we study
here and it would be interesting to investigate further the similarities,
especially from a model-theoretic perspective.

Because contexts in the original Fuzz type system are biased towards the $L^1$
distance, it is not obvious how Fuzz could express the composition principles of
the Hellinger distance.  Recent work showed how this could be amended via a
\emph{path construction} that recasts relational program properties as
sensitivity properties~\cite{DBLP:conf/lics/AmorimGHK19}.  Roughly speaking,
instead of working directly with the Hellinger distance $d_H$, the authors consider a
family of relations $R_{\alpha}$ given by
\[ R_{\alpha} = \{ (\mu_1,\mu_2) \mid d_H(\mu_1, \mu_2) \leq \alpha \}. \] Such
a relation induces another metric on distributions, $d_{\alpha,H}$, where the
distance between two distributions is the length of the shortest path connecting
them in the graph corresponding to $R_\alpha$. This allows them to express the
composition principles of the Hellinger distance directly in the Fuzz type
system, albeit at a cost: the type constructor for probability distributions is
graded by the distance bound $\alpha$.  Thus, the sensitivity information of a
randomized algorithm with respect to the Hellinger distance must also be encoded
in the codomain of the function, as opposed to using just its domain, as done
for the original privacy metric of Fuzz.  By contrast, \system{} does not
require the grading $\alpha$ because it can express the composition principle of
the Hellinger distance directly, thanks to the use of the $L^2$ distance on
bunches.

Duet~\cite{DBLP:journals/pacmpl/NearDASGWSZSSS19} can be seen as an extension of
Fuzz to deal with more general privacy distances. It consists of a two-layer
language: a sensitivity language and a privacy language. The sensitivity
language is very similar to Fuzz. However, it also contains some basic
primitives to manage vectors and matrices. As in Fuzz, the vector types come
with multiple distances but differently from Fuzz, Duet also uses the $L^2$
distance. The main reason for this is that Duet also supports the Gaussian
mechanism which calibrates the noise to the $L^2$ sensitivity of the
function. Our work is inspired by this aspect of Duet, but it goes beyond it by
giving a logical foundation to $L^p$ vector distances.  Another language
inspired by Fuzz is the recently proposed
Jazz~\cite{DBLP:journals/corr/abs-2010-11342}. Like Duet, this language has two
products and primitives tailored to the $L^2$ sensitivity of functions for the
Gaussian mechanism. Interestingly, this language uses contextual information to
achieve more precise bounds on the sensitivities. The semantics of Jazz is
different from the metric semantics we study here; however, it would be
interesting to explore whether a similar contextual approach could be also used
in a metric setting.

\section{Conclusion and Future work}
In this work we have introduced \system, a type system for reasoning
about program sensitivity in the style of
Fuzz~\cite{DBLP:conf/icfp/ReedP10}. \system{} extends the type theory of
Fuzz by considering new type constructors for $L^p$ distances and
bunches to manage different products in typing environments.  We have
shown how this type system supports reasoning about both deterministic
and probabilistic programs.

There are at least two directions that we would like to explore in
future works. On the one hand, we would like to understand if the
typing rules we introduced here could be of more general use in the
setting of probabilistic programs. We have already discussed the
usefulness for other directions in the deterministic
case~\cite{DBLP:series/lncs/6850}. One way to approach this problem
could be by looking at the family of products recently identified
in~\cite{10.1145/3498719}. These products give a model for a logic to reason about
negative dependence between probabilistic variables. It would be
interesting to see if the properties of these products match the one
we have here.

On the other hand, we would like to understand if \system{} can be
used to reason about more general examples in differential
privacy. One way to approach this problem could be to consider examples based on the use of Hellinger distance that have been studied in the literature on probabilistic inference~\cite{DBLP:conf/ccs/BartheFGAGHS16}.

\subsubsection{Acknowledgements} This material is based upon work supported by the NSF under Grant No. 1845803 and 2040249. The third author  was partially supported by the french Program “Investissements d’avenir” (I-ULNE SITE / ANR-16-IDEX-0004 ULNE) managed by the National Research Agency.
\bibliography{biblio}

\ifappendix

\appendix
\section{Term Calculus Proofs}


\bunchmonoidisomorphism*

\begin{proof}
  Proof by structural induction on $\Gamma \leftrightsquigarrow \Delta$. Let $f$ be the inductive hypothesis.

  $\intrp {\Gamma = \Delta} \triangleq \lambda\; \Gamma.\; \Gamma$

  $\intrp {\Gamma_1 ,_p \Gamma_2 \leftrightsquigarrow \Delta_1 ,_p \Delta_2}
    \triangleq \lambda\; (\Gamma_1,\Gamma_2).\; f\; \Gamma_1, f\; \Gamma_2$

  $\intrp{\Gamma \leftrightsquigarrow \cdot ,_p \Delta}
    \triangleq \lambda\; \Gamma.\; ((),f\; \Gamma) $

  $\intrp{\Gamma \leftrightsquigarrow \Delta ,_p \cdot}
    \triangleq \lambda\; \Gamma.\; (f\; \Gamma, ())$

  $\intrp {\Gamma_1 ,_p \Gamma_2 \leftrightsquigarrow \Delta_2 ,_p \Delta_1}
    \triangleq \lambda\; (\Gamma_1,\Gamma_2).\; f\; \Gamma_2, f\; \Gamma_1$

  $\intrp{(\Gamma_1 ,_p (\Gamma_2 ,_p \Gamma_3)) \leftrightsquigarrow ((\Delta_1 ,_p \Delta_2) ,_p \Delta_3)}$

  \;\;\;\;\;\;$\triangleq \lambda\; (\Gamma_1 ,_p (\Gamma_2 ,_p \Gamma_3)).\;
    (f\; \Gamma_1 ,_p f\; \Gamma_2) ,_p f\; \Gamma_3$

  $\intrp{\cdot ,_p \Delta \leftrightsquigarrow \Gamma}
    \triangleq \lambda\; ((), \Delta).\; f\; \Delta$

  $\intrp{\Delta ,_p \cdot \leftrightsquigarrow \Gamma}
    \triangleq \lambda\; (\Delta, ()).\; f\; \Delta$

  $\intrp{((\Delta_1 ,_p \Delta_2) ,_p \Delta_3) \leftrightsquigarrow (\Gamma_1 ,_p (\Gamma_2 ,_p \Gamma_3))}$

  \;\;\;\;\;\;$\triangleq \lambda\; ((\Delta_1 ,_p \Delta_2) ,_p \Delta_3)).\;
    (f\; \Delta_1 ,_p f\; (\Delta_2 ,_p f\; \Delta_3))$
\end{proof}

\supdistance*

\begin{proof}
  It suffices to show that, for all $r \in \Dist$,
  \[\sup_{x\in A} d_B (f(x),g(x)) \le r
    \iff
    \sup_{x,y\in A} d_B (f(x),g(y)) - d_A(x,y) \le r. \]
  \begin{enumerate}
    \item $(\implies)$

    By the triangle inequality we know
    $$d_B(f(x), g(y)) \le d_B(f(x), g(x)) + d_B(g(x), g(y))$$
    and non-expansiveness gives the inequality
    $$d_B(g(x), g(y)) \le d_A(x, y)$$
    so we can subtract from both sides and get:
    $$d_B(f(x), g(y)) - d_A(x, y) \le d_B(f(x), g(x)) \le r$$

    \item $(\impliedby)$

    By the definition of $\sup$ we get
    $$\sup_{x,x\in A} d_B (f(x),g(x)) - d_A(x,x) \le
    \sup_{x,y\in A} d_B (f(x),g(y)) - d_A(x,y) \le r$$
    Now simplifying the left hand side we know $d_A(x,x) = 0$ because of identity, so $$\sup_{x\in A} d_B (f(x),g(x)) \le r$$
  \end{enumerate}
\end{proof}

\lollisubtyping*

\begin{proof}
  Let $r_1$ be the distance between $f$ and $g$:
  $$
    d_{A \multimap_1 B}(f,g) =
      \arginf_{r_1\in \R^+\cup \{\infty\}}
        \forall x,y \in A, d_B(f(x), g(y)) \le r_1 + d_A(x,y)
  $$
  and let $r_2$ be the distance between $f$ and $g$:
  $$
    d_{A \multimap_p B}(f,g) =
      \arginf_{r_2\in \R^+\cup \{\infty\}}
        \forall x,y \in A, d_B(f(x), g(y)) \le \sqrt[p]{r_2^p + d_A(x,y)^p }
  $$
  Because each $\arginf$ of $r_1$ and $r_2$ are both minimizing to the same value: $\forall x,y \in A, d_B(f(x), g(y))$, we can set that constant and say they are minimizing to the same constant $c$. Also because the smallest number less than or equal to a constant $c$ is $c$, we know that each $\arginf$ is minimizing the expressions $r_1 + d_A(x,y)$ and $\sqrt[p]{r_2^p + d_A(x,y)^p }$ such that they equal $c$. This means we have
    $$r_1 + d_A(x,y) = \sqrt[p]{r_2^p + d_A(x,y)^p }$$
  Using two properties of the $L^p$ metric we find that $r_1 \le r_2$. First by the well-ordered property of the $L^p$ metric,
  $$\forall x, y, ||(x, y)||_{p} \ge ||(x,y)||_{p + 1}$$
  So $r_1$ and $r_2$ must vary because $d_A(x,y)$ is constant. The $L^p$ metric is also monotone with regards to its arguments, meaning that if $x_1 \le x_2$ then $||(x_1, y)||_{p} \le ||(x_2, y)||_p$.

  So given that $||(r_1, d_A(x,y))||_1 = ||(r_2,d_A(x,y))||_p$ then we know that $r_1 \le r_2$ to compensate for the well-ordered property. And because $r_1$ and $r_2$ are the distances returned from $d_{A \multimap_p B}(f,g)$ our lemma holds.
\end{proof}

\bindsoundness*

\begin{proof}
  By unfolding the definitions of non-expansiveness and applying standard
  results about the Hellinger distance.  We focus on bind. The composition
  principle for the Hellinger metric as defined in
  \cite{DBLP:conf/icalp/BartheO13} Proposition 5 is, for $\mu, \nu \in DA$ and
  $f, g \in A \to DB$:
$$
\HD_B(\lift f \mu, \lift g \nu) \le \sqrt{\HD_A(\mu,\nu)^2 + \sup_{x\in A} \HD_B(f(x), g(x))^2}
$$
This shows that the semantics of bind is a non-expansive map. With some algebraic
manipulation we can see the Hellinger distance satisfies
\begin{align*}
  \HD_B(\lift f \mu, \lift g \nu)
    &\le \sqrt{\HD_A(\mu,\nu)^2 + \sup_{x\in A} \HD_B(f(x), g(x))^2} \\
    &\le ||(d_{\bigcirc_H A}(\mu,\nu), \sup_{x\in A} \HD_B(f(x), g(x)))||_2 \\
    &\le ||(d_{\bigcirc_H A}(\mu,\nu), d_{\multimap_1}(f,g))||_2
    & \text{(by \Cref{thm:d-iff-dhat}).} \\
\end{align*}
Hence the type of bind can be expressed as:
$$
(!_sA \multimap_1 \bigcirc_H B)  \otimes_2 \bigcirc_H A \longrightarrow \bigcirc_HB.
$$
We obtain the sought type by applying \Cref{thm:lolli-1-le-lolli-p}.
\end{proof}

\soundness*

\begin{figure*}
\begin{mathpar}
  \inferrule*[Right=Contr]{
  \Gamma \approx \Delta
  \and
  \inferrule*[Right=\tensorI]{
  \Gamma \vdash e_1 : \bigcirc_H \tau
  \and
  \inferrule*[Right=\lolliI]{
  \inferrule*[Right=!E]{
    \inferrule*[Right=!I]{
      \inferrule*[Right=Axiom]{ }{[x : \tau]_1 \vdash x : \tau}
    }{[x : \tau]_s \vdash !x : \!_s \tau}
    \and
    \Delta,_p [x : \tau]_s \vdash e_2 : \bigcirc_H \sigma
  }{\Delta,_p [x : \!_s\tau]_1 \vdash \LetBangIn{x}{x}{e_2} : \bigcirc_H \sigma}
  }{\Delta \vdash \lambda x.\LetBangIn{x}{x}{e_2} : {!_s \tau} \multimap_p \bigcirc_H \sigma}
  }{\Gamma ,_2 \Delta \vdash (e_1, \lambda x. \LetBangIn{x}{x}{e_2}) : \bigcirc_H \tau \otimes_2 (!_s A \multimap_p  \bigcirc_H \sigma)}
  }{Contr(2, \Gamma, \Delta) \vdash (e_1, \lambda x. \LetBangIn{x}{x}{e_2}) : \bigcirc_H \tau \otimes_2 (!_s \tau \multimap_p  \bigcirc_H \sigma)}
\end{mathpar}
\caption{Derivation for soundness of bind}\label{fig:bind-soundness}
\end{figure*}

\begin{proof}
  Every inductive step in \Cref{def:derivation-semantics} is independently non-expansive, and non-expansive functions combine to create non-expansive functions. Hence our semantics is sound. Let's look at some key cases. Consider the \lolliE{} case: (variable names have been slightly altered to avoid confusion)
  \begin{mathpar}
    \inferrule*[Right=\lolliE]
      {\Gamma \vdash f : A \multimap_p B \and \Delta \vdash e : A}
      {\Gamma ,_p \Delta \vdash f \: e : B}
    \and
    \intrp {\text{\lolliE} \; \pi_1 \; \pi_2} \triangleq \lambda (\gamma, \delta).\; \intrp {\pi_2} \; \gamma \; ( \intrp{\pi_1} \; \delta )
  \end{mathpar}
  The \lolliE{} rule takes two derivations $\pi_1$ and $\pi_2$ of types $A \multimap_p B$ and $A$ respectively. $\intrp{\pi_1}$ is a non-expansive function in the set $\intrp\Gamma \to \intrp A \to \intrp B$. $\intrp{\pi_2}$ is a non-expansive function in the set $\intrp\Delta \to \intrp A$. We want to create a non-expansive function of type $\intrp{\Gamma,_p\Delta} \to \intrp B$ which expands to $\intrp{\Gamma}\times\intrp{\Delta} \to \intrp B$. So the function we make is from a pair of environments to a $\intrp B$. $\intrp {\pi_2} \delta$ applies the interpretation of $\pi_2$ to $\delta$ to get an element of $\intrp A$. This is then used as an argument for $\intrp {\pi_2}$ to get an element of $\intrp B$. All functions are non-expansive so the created function is also non-expansive.
  As a simpler example, consider the $\!$I rule.
  \begin{mathpar}
    \inferrule*[Right=!I]
      {\Gamma \vdash e : A}
      {s\Gamma \vdash \! e : \!_s A}
    \and
    \intrp {!I \; \pi} \triangleq \intrp{\pi}
  \end{mathpar}
  This is non-expansive because enforcing sensitivity constraints happens at the type level and is used in the distance metric for the set. Recall that the distance metric for $\intrp{!A}$ is $d_{!A}(x, y) = s\cdot d_A(x, y)$. So the carrier set for the type $\intrp{!_sA}$ is just $\intrp{A}$.

  To restate the proof, all individual steps in the proof are non-expansive, which compose into larger non-expansive functions and so the semantics are sound.
\end{proof}

\begin{figure}
  \begin{mathpar}
    \inferrule*[Right=\lolliI]{
    \inferrule*[Right=\lolliI]{
    \inferrule*[Right=Exch]{
    \inferrule*[Right=\lolliE]{
      \inferrule*[Right=Ax]
      { }
      {[f : (A \otimes_p B) \multimap_p B]_1 \vdash f : (A \otimes_p B) \multimap_p B}
      \and
      \inferrule*[Right=\tensorI]
      {
        \inferrule*[Right=Ax]{ }{[a : A]_1 \vdash a : A}
        \and
        \inferrule*[Right=Ax]{ }{[b : B]_1 \vdash b : B}
      }
      {[a : A]_1,_p[b : B]_1 \vdash (a, b) : A \otimes_p B}
    }
    {
      [f : (A \otimes_p B) \multimap_p B]_1, ([a : A]_1,_p[b : B]_1)
        \vdash f (a, b) : C
    }}
    {
      ([f : (A \otimes_p B) \multimap_p B]_1, [a : A]_1),_p[b : B]_1
        \vdash f (a, b) : C
    }}
    {
      [f : (A \otimes_p B) \multimap_p B]_1, [a : A]_1
        \vdash \lambda b. f (a, b) : B \multimap_p C
    }}
    {
      [f : (A \otimes_p B) \multimap_p B]_1
        \vdash \lambda a. \lambda b. f (a, b) : A \multimap_p B \multimap_p C
    }
    \and
    \inferrule*[Right=\lolliI]{
      \inferrule*[Right=\tensorE]{
        \inferrule*[Right=Ax]{ }{[x : A \otimes_p B]_1 \vdash x : A \otimes_p B}
        \and
        \inferrule*[Right=Exch]{
          \inferrule*[Right=\lolliE]{
            \inferrule*[Right=\lolliE]{
              \inferrule*[Right=Ax]{ }{[f : A \multimap_p B \multimap_p C]_1 \vdash f : A \multimap_p B \multimap_p C}
              \and
              \inferrule*[Right=Ax]{ }{[a : A]_1 \vdash a : A}
            }{
              [f : A \multimap_p B \multimap_p C]_1 ,_p [a : A]_1 \vdash f a : B \multimap_p C
            }
            \and
            \inferrule*[Right=Ax]{ }{[b : B]_1 \vdash b : B}
          }{
            ([f : A \multimap_p B \multimap_p C]_1 ,_p [a : A]_1) ,_p [b : B]_1
              \vdash f\;a\;b : C
          }}{
          [f : A \multimap_p B \multimap_p C]_1 ,_p ([a : A]_1 ,_p [b : B]_1)
            \vdash f\;a\;b : C
        }}{
        [f : A \multimap_p B \multimap_p C]_1 ,_p [x : A \otimes_p B]
          \vdash \LetPairIn{a}{b}{x}{f\;a\;b} : C
      }}{
      [f : A \multimap_p B \multimap_p C]_1
        \vdash \lambda x.\; \LetPairIn{a}{b}{x}{f\;a\;b} : A \otimes_p B \multimap_p C
    }
    \end{mathpar}
  \caption{Derivation of currying and uncurrying}\label{fig:curry-deriv}
\end{figure}

\section{Extra Examples}\label{sec:extra_example}

\paragraph{Zero and Infinity}\label{zero_and_infty}
The choice of $\infty \cdot 0 = 0 \cdot \infty = \infty$ is a careful one to avoid bugs and preserve soundness. We are using the same behavior as \cite{DBLP:conf/lics/AmorimGHK19}. Another possible definition would be that of \cite{DBLP:conf/popl/AmorimGHKC17}. To see why this must be the behavior of multiplying zero and infinity, consider the following Fuzz program.

\begin{minipage}[t]{0.3\linewidth}
  \begin{verbatim}
    if x < y
    then 1
    else 0
  \end{verbatim}
\end{minipage}
\begin{minipage}[t]{0.3\linewidth}
  \vspace*{0.4cm}
  which desugars to
\end{minipage}
\begin{minipage}[t]{0.3\linewidth}
  \begin{verbatim}
    case x < y of
    | inl () -> 1
    | inr () -> 0
    end
  \end{verbatim}
\end{minipage}

$x$ and $y$ should be marked as $\infty$-sensitive because we are using the $<$
operation with them, however the body of the case is zero sensitive to the
value returned by $x < y$ so the sensitivity of the expression will be:
$\infty \cdot 0 = \infty$

\paragraph{Rotations}

\aaa{Can we say anything when we vary $\theta$?}

As a warm-up, let us consider how we can extend \system{} with a primitive for
computing rotations on the Cartesian plane.  Given a rotation angle
$\theta \in \R$, we define the following function $R_\theta$:
\begin{align*}
  R_{\theta} & : \R^2 \to \R^2 \\
  R_\theta(x,y) & =(\cos(\theta)x- \sin(\theta)y,\sin(\theta)x+\cos(\theta)y).
\end{align*}
Using the $L^2$ distance we have, for any $(x,y), (x',y') \in \mathbb R^2 $:
$$ d_{2}( R_{\theta}(x,y), R_{\theta}(x',y'))=d_{2}( (x,y),
(x',y')).$$ So, as a function on $(\R^2, d_2)$, $R_{\theta}$ is non-expansive.
In other words, it has type $\R \otimes_2 \R \multimap \R \otimes_2 \R$.

Note that, by contrast, $R_{\theta}$ is not non-expansive for the $L^1$ or
$L^\infty$ distances.  For instance, suppose that $\theta = \pi / 4$, and let
$p = (\sqrt{2}/2,\sqrt{2}/2)$. Then
\begin{align*}
  R_\theta(0,0) & = (0,0) &
  R_\theta(1,0) & = p &
  R_\theta(p) & = (0,1).
\end{align*}
Thus, $d_{1}(R_\theta(0,0),R_\theta(1,0)) = \sqrt{2}/2 + \sqrt{2}/2 = \sqrt{2}$,
which is strictly larger than $d_{1}((0,0),(1,0)) = 1$.  Similarly,
$d_{\infty}(R_\theta(0,0),R_\theta(p)) = \max(0,1) = 1$, which is strictly
larger than $d_{\infty}((0,0),p) = \max(\sqrt{2}/2,\sqrt{2}/2) = \sqrt{2}/2$.

\paragraph{Computing distances}

Suppose that the type $\tau$ denotes a proper metric space.  Then we can
incorporate its distance function in \system{} with the type
\[ \tau \otimes_1 \tau \multimap \R. \] Indeed, let $x$, $x'$, $y$ and $y'$ be
arbitrary elements of $\intrp \tau$.  Then
\begin{align*}
  d(x,y) - d(x',y')
  & \leq d(x,x') + d(x',y') + d(y',y) - d(x',y') \\
  & = d(x,x') + d(y,y') \\
  & = d_{\intrp{\tau \otimes_1 \tau}}((x,y),(x',y')).
\end{align*}
By symmetry, we also know that
$d(x',y') - d(x,y) \leq d_{\intrp{\tau\otimes_1 \tau}}((x,y),(x',y'))$.
Combined, these two facts show
\begin{align*}
  d_{\R}(d(x,y), d(x',y'))
  & = |d(x,y) - d(x',y')| \\
  & \leq d_{\intrp{\tau \otimes_1\tau}}((x,y),(x',y')),
\end{align*}
which proves that the metric on $\intrp{\tau}$ is indeed a non-expansive
function.

\paragraph{Distributivity properties}

In linear logic the following distributivity properties are derivable:
\begin{align*}
  A \otimes (B \oplus C)  & \vdash (A \otimes B) \oplus (A \otimes C) \\
  (A \otimes B) \oplus (A \otimes C)  & \vdash A \otimes (B \oplus C) \\
  (A \with B) \oplus (A \with C) & \vdash A \with (B \oplus C).
\end{align*}
However, $\with$ does not distribute perfectly over $\oplus$, since the converse of
the last statement does not usually hold:
\begin{align*}
  A \with (B \oplus C)   \not\vdash (A \with B) \oplus (A \with C).
\end{align*}

By contrast, in \system{}, $\otimes_p$ does distribute over $\oplus$, as
witnessed by the following judgments
\begin{align*}
[u : \tau \otimes_p (\sigma \oplus \rho)]_1 & \vdash t_1 : (\tau \otimes_p \sigma) \oplus (\tau \otimes_p \rho)\\
[u : (\tau \otimes_p \sigma) \oplus (\tau \otimes_p \rho)]_1 & \vdash t_2 : \tau \otimes_p (\sigma \oplus \rho),
\end{align*}
where
\begin{align*}
  t_1 = \LetPairIn {&u_1} {u_2} {u} {} \text{case } u_2 \text{ of}\\
    &|\;x.\;\inj_1(u_1,x)\\
    &|\;y.\;\inj_2(u_1,y)\\
  t_2 = \textbf{case }& u \textbf{ of}\\
    &|\; x.\;\LetPairIn{x_1}{x_1}{x}{(x_1, \inj_1 x_2)}\\
    &|\; y.\;\LetPairIn{y_1}{y_1}{y}{(y_1, \inj_2 y_2)}
\end{align*}

We can also show the following distributivity properties of scaling:
\begin{align*}
[x : \tau \otimes_p \sigma]_r & \vdash t_1 : {!_r \tau} \otimes_p {!_r \sigma}\\
[z : \!_r \!_s\tau]_1 & \vdash t_2 : \!_s \!_r \tau
\end{align*}
where $t_1$ and $t_2$ are
\begin{align*}
  t_1 & = \LetPairIn {x_1} {x_2} {x} {(!_r x_1,\!_r x_2)}\\
  t_2 & = \LetBangIn{y}{z}{ \LetBangIn{w}{y}{!! w }  }.
\end{align*}

\paragraph{Programming with matrices}

The Duet language~\cite{DBLP:journals/pacmpl/NearDASGWSZSSS19} provides several
matrix types with the $L^1$, $L^2$, or $L^\infty$ metrics, along with primitive
functions for manipulating them.  In \system{}, these types can be defined
directly as follows
\[ \mathbb{M}_p[m,n] = \otimes_1^m \otimes_p^n \R. \] Following Duet, we use the
$L^1$ distance to combine the rows and the $L^p$ distance to combine the
columns.
One advantage of having types for matrices defined in terms of more basic
constructs is that we can program functions for manipulating them directly,
without resorting to separate primitives.  For example, we can define the
following terms in the language:
\begin{align*}
addrow &: \mathbb{M}_{p} [1,n] \otimes_1 \mathbb{M}_{p} [m,n]
\multimap \mathbb{M}_{p} [m+1,n]  \\
addcolumn &: \mathbb{M}_{1} [1,m] \otimes_1 \mathbb{M}_{1} [m,n]
\multimap \mathbb{M}_{1} [m,n+1] \\
addition &: \mathbb{M}_{1} [m,n] \otimes_1 \mathbb{M}_{1} [m,n] \multimap \mathbb{M}_{1} [m,n].
\end{align*}
The first program, $addrow$, appends a vector, represented as a $1 \times n$
matrix, to the first row of a $m\times n$ matrix. The second program,
$addcolumn$, is similar, but appends the vector as a column rather than a row.
Because of that, it is restricted to $L^1$ matrices.
\aaa{Can we generalize addcolumn to $L^p$ matrices? It feels like it should be
  possible, at least if we use $\otimes_p$ to separate the two arguments.}
Finally, the last program, $addition$, adds the elements of two matrices
pointwise.

 One drawback of our encoding is that these programs need to be defined
 separately for each matrix dimension.  In practice, it would be desirable to
 have a dependently typed version of \system{}, along the lines of
 DFuzz~\cite{DBLP:conf/popl/GaboardiHHNP13}, to simplify the manipulation of
 matrices of arbitrary size.

\paragraph{Metrics for lists and inductive types}

In Fuzz, we can define two list types using
recursion~\cite{DBLP:conf/icfp/ReedP10}:
\begin{align*}
\mathtt{list} \; \tau & = \mu \alpha. 1 \oplus (\tau \otimes \alpha) &
\mathtt{alist} \; \tau & = \mu \alpha. 1 \oplus (\tau \with \alpha).
\end{align*}
Following prior work~\cite{DBLP:conf/popl/AmorimGHKC17}, these types can be
interpreted as metric spaces, by computing the initial algebra of a certain
functor.  The carrier of these metric spaces is the set of lists over
$\intrp{\tau}$, endowed with the following metrics:
\begin{align*}
  d_{list}(l,l') & =
  \begin{cases}
    \infty & \text{if $|l| \neq |l'|$} \\
    \sum_{i=1}^{n} d_A(l_i,l'_i) & \text{if $|l| = |l'| = n$}
  \end{cases} \\
  d_{alist}(l,l') & =
  \begin{cases}
    \infty & \text{if $|l| \neq |l'|$} \\
    \max_{i=1}^{n} d_A(l_i,l'_i) & \text{if $|l| = |l'| = n$}.
  \end{cases}
\end{align*}
This construction can be easily adapted to \system and generalized.  First, we
extend \system{} with inductive types, by which we mean recursive types with
strictly positive recursive occurrences (dealing with arbitrary recursive types
should be possible by using a variant of metric
CPOs~\cite{DBLP:conf/popl/AmorimGHKC17}).  Then, we define
\[p\mathtt{list} \; \tau = \mu \alpha. 1 \oplus (\tau \otimes_p \alpha).\] To
interpret such inductive types, we follow the standard recipe.  First, by
standard categorical arguments, we can show that the category of metric spaces
and non-expansive functions has colimits of chains.  Specifically, given a chain
$X_i$ of metric spaces, we can define $X_\infty = \colim_i X_i$ via the formula
\begin{align*}
  |X_\infty| & = \colim_i |X_i| \\
  d_{X_\infty} & = \inf_i d_{X_i}^*,
\end{align*}
where $d_{X_i}^*$ denotes the pushforward of the metric $d_{X_i}$ into
$|X_\infty|$.  Second, we note that a type expression $\tau$ with one free type
variable $\alpha$ corresponds to a cocontinuous functor on metric spaces,
because it is formed by composing cocontinuous functors.  We can compute the
initial algebra of this functor as the colimit of a certain chain, which we take
to be the interpretation of $\mu \alpha. \tau$.

In the case of $p\mathtt{list}\;\tau$, by unfolding definitions, we obtain the
following metric:
\begin{align*}
  d_{\plist\tau}(l,l') & =
  \begin{cases}
    \infty & \text{if $|l| \neq |l'|$} \\
    \sqrt[p]{\sum_{i=1}^{n} d_A(l_i,l'_i)} & \text{if $|l| = |l'| = n$}.
  \end{cases}
\end{align*}
In the cases $p \in \{1,\infty\}$, this reduces to the previous distances on
lists (where, in the case $p = \infty$, we take the limit of the right-hand side
when $p \to \infty$).

The $\plist\; \tau$ type is equipped with the following constructors:
\begin{align*}
nil &: \plist \tau \\
cons &:  \tau \otimes_p \plist\tau \multimap \plist \tau.
\end{align*}
Moreover, we can define functions on lists by structural recursion, which we can
soundly add to \system{} thanks to the universal property of initial algebras.
For example:
\begin{align*}
  append &: \plist \tau  \otimes_p \plist \tau \multimap \plist \tau.
\end{align*}

\paragraph{K-Means} \label{ex:k-means} The k-means algorithm is an iterative algorithm for
finding multiple means in a set of datapoints. These means can be
thought of as approximate "centers" of groupings in the dataset. A
differentially private version of the k-means algorithm typed in Fuzz had been given in \cite{DBLP:conf/icfp/ReedP10}, using the Laplace mechanism.
Here we revisit this example to
illustrate how by using \system{}  typing and $L^p$  one can refine the
sensitivity analysis of an algorithm.

 The $iterate$ Fuzz term defined in \cite{DBLP:conf/icfp/ReedP10}
 takes a set of data points, a list of centers and returns an updated
 list of centers, obtained by grouping each data point to the center
 it is closest to, adding Laplacian noise and then taking the new
 centers to be the mean of each group. It was given the following Fuzz
type\footnote{Actually there were two typos on types in
  \cite{DBLP:conf/icfp/ReedP10}; the type of $iterate$ given here is
the right corrected one, as well as the type of $zip$.}
$$ iterate : \!_{3} (\set pt) \multimap  !_{\infty}(\Fuzzlist pt) \multimap    \bigcirc_P ( \Fuzzlist pt)$$
The 3 sensitivity of $iterate$ in its first data points set argument comes from
the fact that this argument is used 3 times in the term, thus the Fuzz
contraction rule (in $L^1$) leads to an index 3 for the $!$ of this
argument. The idea here is to use instead in \system{} contraction in a $L^p$
bunch context, which will lead to an index $\sqrt[p]3$ instead of 3. For
enabling that one needs to change the type of the $zip$ intermediate function,
replacing $\multimap$ with $\multimap_p$. Then this forces to take for points
the type $pt=\R \otimes_p \R$ (so using the \lp metric), for lists the type
$\plist \tau$, and to change accordingly the type of $map$ and of the other
intermediary functions. We obtain the following types:








\begin{align*}
  pt & \triangleq \R \otimes_p \R\\
assign &:  !_{\infty}(\plist pt )\multimap_p \set pt \multimap_p \set (pt \otimes_p int)
  \\partition &: \set (pt \otimes_p int) \multimap_p \plist (\set pt)
  \\totx, toty &: \set pt \multimap \R
  \\zip &: \plist \tau \multimap_p \plist \sigma \multimap_p \plist (\tau \otimes_p \sigma)
  \\pmap &: \!_\infty(\tau \multimap_p \sigma) \multimap_p (\plist \tau) \multimap_p \plist \sigma
\end{align*}

The term $assign$ takes a list of means and a database and returns pairs of
points matched with the index of the closest mean given in the list of
means. $partition$ takes these labeled points and splits them into a list of
sets of points. $totx$ and $toty$ calculate the total of the $x$ or $y$
coordinates respectively in a set of points. $zip$ is the usual zip function on
lists, and $pmap$ is the usual map function adapted to our fixed \lp
space. Finally, $seq$ binds over every element in a list to take a list of
distributions and return a distribution over
lists. 

\begin{align*}
  &seq : \plist (\bigcirc_P \tau) \multimap_p \bigcirc_P (\plist \tau)\\
  &seq\; [] = []\\
  &seq\; x :: xs =
  \mLet{y}{x}{
    \mLet{ys}{seq\; xs}{
      \return\; y :: ys}}
\end{align*}
The k-means algorithm is defined below. The user supplies a database and a set of $k$ initial means. The means are either initialized to random points within the dataset or are the output of a previous iteration of the algorithm. The datapoints are then grouped by distance to each mean using $assign$ and new means are calculated by taking the average of each groups $x$'s and $y$'s.
\begin{align*}
 iterate : \!_{\sqrt[p] 3} (&\set pt) \multimap_p !_{\infty}(\plist pt) \multimap_p
  \bigcirc_P (\plist pt)\\
  iterate\; b\; ms =\;&\LetBangIn{b'}{b}{
    \\ &\LetIn{b''}{partition\; (assign\; ms\; b')}{
      \\ &\LetIn{tx}{pmap (add\_noise \circ totx)\; b''}{
        \\ &\LetIn{ty}{pmap (add\_noise \circ toty)\; b''}{
          \\ &\LetIn{t}{pmap (add\_noise \circ size)\; b''}{
            \\ &\LetIn{stats}{zip\; (zip\; (tx, ty), t)}{
              \\ &seq\; (pmap\; avg\; stats)
            }
          }
        }
      }
    }
  }
\end{align*}

Note that if we take $p=1$ we have exactly the same type derivation as
in \cite{DBLP:conf/icfp/ReedP10} in Fuzz.

It is also possible to write a variant of this program which instead
of building two lists, one for component $x$ and one for component
$y$, builds a single lists of vectors  in $\R \otimes_p \R$ by using
a map on the function $\mathit{vectorSum}$ defined before. The sensitivity
obtained with respect to the set of data points argument  is then $1+2^{1/p}$. For this variant one only uses Bunched Fuzz
connective  $\otimes_p$ for the underlying vector type $\R \otimes_p
\R$ but one keeps the Fuzz versions (with $\multimap$) of $map$, $zip$, $assign$\dots
The noise is also added by the Laplace mechanism.



\contrsound*

\begin{proof}
  By induction on the derivation of $\Gamma \approx \Delta$.  The first point is
  trivial, since $\approx$ relates bunches that differ only on variable names
  and sensitivities, which do not affect the carrier sets.  Thus, we focus on
  the last point.  The case $p = \infty$ is easier, since in this case
  $Contr(\infty, -, -)$ takes the pointwise maximum of all the sensitivities in
  the contexts, and because $\otimes_\infty$ becomes a true product in the
  categorical sense.  Now, suppose that $p < \infty$.
  \begin{itemize}
  \item If $\Gamma$ and $\Delta$ are empty, then the domain and codomain of
    $\delta$ is reduced to a singleton set.  Thus, $\delta$ is trivially
    non-expansive.
  \item Now suppose that $\Gamma = [x : \tau]_s$ and $[y : \tau]_r$.  We need to
    show that the diagonal function is a non-expansive function of type
    \[ !_{\sqrt[p]{s^p +r^p}} \intrp\tau \to {!_s \intrp\tau} \otimes_p {!_r
        \intrp\tau}. \]
    Let $X$ denote the domain of this map, and $Y$ the
    codomain.  First, suppose that $p < \infty$.  Non-expansiveness holds
    because
    \begin{align*}
      d_Y((x,x),(y,y))
      & = \sqrt[p]{(s \cdot d(x,y))^p+(r \cdot d(x,y))^p} \\
      & = \sqrt[p]{(s^p+r^p)d(x,y)^p} \\
      & = \sqrt[p]{s^p+r^p}d(x,y) \\
      & = d_X(x,y).
    \end{align*}
    If $p = \infty$, the above root is actually defined as $\max(s,r)$.  In this
    case, we have
    \begin{align*}
      d_Y((x,x),(y,y))
      & = \max(s \cdot d(x,y), r \cdot d(x,y)) \\
      & \leq \max(\max(s,r) \cdot d(x,y), \max(s,r) \cdot d(x,y)) \\
      & = \max(s,r) \cdot d(x,y) \\
      & = d_X(x,y).
    \end{align*}
  \item Now suppose that $\Gamma = \Gamma_1 ,_q \Gamma_2$,
    $\Delta = \Delta_1 ,_q \Delta_2$, $\Gamma_1 \approx \Delta_1$ and
    $\Gamma_2 \approx \Delta_2$. Abbreviate
    $c(p,q) = 2^{\left|\frac1p - \frac1q\right|}$ as just $c$. By induction, the
    diagonals are non-expansive functions of types
    \begin{align*}
      \intrp{Contr(p,\Gamma_1,\Delta_1)}
      & \to \intrp{\Gamma_1}\otimes_p\intrp{\Delta_1} \\
      \intrp{Contr(p,\Gamma_2,\Delta_2)}
      & \to \intrp{\Gamma_2}\otimes_p\intrp{\Delta_2}.
    \end{align*}
    We can rewrite the diagonal on $\intrp{Contr(p,\Gamma,\Delta)}$ as the
    composite
    \begin{align*}
      & \intrp{Contr(p,\Gamma,\Delta)}\\
      & = {!_c\intrp{Contr(p,\Gamma_1,\Delta_1)}} \otimes_{q}
        {!_c\intrp{Contr(p,\Gamma_2,\Delta_2)}}
      & \text{\Cref{prop:scaled-bunch-comm}, def.} \\
      & \to {!_c(\intrp{\Gamma_1} \otimes_{p} \intrp{\Delta_1})} \otimes_{q}
        {!_c(\intrp{\Gamma_2} \otimes_{p} \intrp{\Delta_2})}
      & \text{induction} \\
      & = {!_c((\intrp{\Gamma_1} \otimes_{p} \intrp{\Delta_1}) \otimes_{q}
        (\intrp{\Gamma_2} \otimes_{p} \intrp{\Delta_2}))}
      & \text{\Cref{prop:scaling-distr}} \\
      & \to (\intrp{\Gamma_1} \otimes_{q} \intrp{\Gamma_2})
        \otimes_{p} (\intrp{\Delta_1} \otimes_{q} \intrp{\Delta_2})
      & \text{\Cref{prop:tensor-distr}} \\
      & = \intrp{(\Gamma_1,_q\Gamma_2)} \otimes_p
        \intrp{(\Delta_1,_q\Delta_2)}\\
      & = \intrp\Gamma \otimes_p \intrp\Delta.
    \end{align*}
  \end{itemize}
\end{proof}

\ifalgrules
\section{Algorithmic Rules}


The system of algorithmic rules is displayed on Fig. \ref{fig:alg-rules}.
\begin{figure}
  \begin{mathpar}
    \inferrule*[Right=Axiom]
      {s \ge 1}
      {[x : \tau]_s \vdash x : \tau}
    \and
    \inferrule*[Right=\RI]
      { }
      {\cdot \vdash r : \R}
    \and
    \inferrule*[Right=1I]{ }{\cdot \vdash () : 1}
    \\
    \inferrule*[Right=\lolliI alg]
      {\Gamma \,_p [x : \tau]_1 \vdash e : \sigma}
      {\Gamma \vdash \lambda x. e : \tau \multimap_p \sigma}
    \and
    \inferrule*[Right=\lolliE alg]
      {\Gamma \vdash f : \tau \multimap_p \sigma \and \Delta \vdash e : \tau \and \Gamma \approx \Delta}
      {Contr(p, \Gamma, \Delta) \vdash f \: e : \sigma}
    \\
    \inferrule*[Right=\tensorI alg]
      {\Gamma \vdash e_1 : \tau \and \Delta \vdash e_2 : \sigma \and \Gamma \approx \Delta}
      {Contr(p, \Gamma , \Delta) \vdash (e_1, e_2) : \tau \otimes_p \sigma}
    \and
    \inferrule*[Right=\tensorE alg]
      { \Delta \vdash e_1 : \tau \otimes_p \sigma
        \and
        \Gamma,_q([x : \tau]_s \,_p [y : \sigma]_s) \vdash e_2 : \rho
        \and
        \Gamma \approx \Delta }
      { Contr(p, \Gamma, s\Delta) \vdash
        \textbf{let$_q$ } (x,_p y) = e_1 \textbf{ in } e_2
        : \rho}
    \\
    \inferrule*[Right=\plusI1 alg]
      {\Gamma \vdash e : \tau}
      {\Gamma \vdash \inj_1 e : \tau \oplus \sigma}
    \and
    \inferrule*[Right=\plusI2 alg]
      {\Gamma \vdash e : \sigma}
      {\Gamma \vdash \inj_2 e : \tau \oplus \sigma}
    \and
    \inferrule*[Right=\plusE alg]
      {
        \Gamma \vdash e_1 : \tau \oplus \sigma
        \and \Delta,_p [x : \tau]_s \vdash e_2 : \rho
        \and \Delta,_p [y : \sigma]_s \vdash e_3 : \rho
        \and \Gamma \approx \Delta }
      {Contr(p, \Delta, s\Gamma) \vdash
        \textbf{case$_p$ } e_1 \textbf{ of } x.\;e_3 \mid y.\;e_3 }
    \\
    \inferrule*[Right=!I]
      {\Gamma \vdash e : \tau}
      {s\Gamma \vdash \! e : \!_s \tau}
    \and
    \inferrule*[Right=!E alg]
      {\Gamma \vdash e_1 : \!_r \tau \and \Delta,_p[x : \tau]_{rs} \vdash e_2 : \sigma}
      {Contr(p, \Delta, s\Gamma) \vdash \textbf{let$_p$ } !x = e_1 \textbf{ in } e_2 : \sigma}
    \\
    \inferrule*[Right=Bind-P]
      { \Gamma \vdash e_1 : \bigcirc_P \tau
        \and
        \Delta ,_p [x : \tau]_s \vdash e_2 : \bigcirc_P \sigma
        \and
        \Gamma \approx \Delta}
      {Contr(1, \Gamma, \Delta) \vdash
        \textbf{mlet$_p$ } x = e_1 \textbf{ in } e_2}
    \and
    \inferrule*[Right=Return-P]
      { \Gamma \vdash e : \tau}
      { \infty\Gamma \vdash \return e : \bigcirc_P \tau}
    \\
    \inferrule*[Right=Bind-H]
      { \Gamma \vdash e_1 : \bigcirc_H \tau
        \and
        \Delta ,_p [x : \tau]_s \vdash e_2 : \bigcirc_H \sigma
        \and
        \Gamma \approx \Delta}
      {Contr(2, \Gamma, \Delta) \vdash
        \textbf{mlet$_p$ } x = e_1 \textbf{ in } e_2}
    \and
    \inferrule*[Right=Return-H]
      { \Gamma \vdash e : \tau}
      { \infty\Gamma \vdash \return e : \bigcirc_H \tau}
  \end{mathpar}
  \caption{Algorithmic Rules}\label{fig:alg-rules}
\end{figure}

\fi

\section{The language as a logic}\label{sect:logic}

%

We give an alternative presentation of the non-probabilistic fragment
of \system{} as a logic,
by means of a sequent calculus. This logic shares many of the properties
of \system{}. We have also proved a cut elimination result for it.

Bunches with multiple holes labeled by a set of variables $X$ are denoted with
$\Gamma\{ x \mapsto \star \}_{x \in X}$.

\paragraph{Formulas} The syntax of formulas follows much of the same structure as \system{}'s type system.
\begin{align*}
  &A, B ::= 1
    \mid \bot
    \mid \mathbb R
    \mid \!_sA
    \mid A \multimap_p B
    \mid A \otimes_p B
    \mid A \oplus B
  \\ & p \in \R^{\ge1}_\infty , s \in \Sens
\end{align*}
\paragraph{Bunches} Environments are defined as
$$\Gamma ::= \cdot \mid [ A ]_s \mid \Gamma \,_p \Gamma$$
and enjoy the same properties as in \system.

\begin{figure*}
  \begin{mathpar}
  \inferrule*[Right=Axiom]
    { }
    {[A]_1 \vdash A}
  \and
  \inferrule*[Right=\RR]
    { }
    {\cdot \vdash \mathbb R}
  \and
  \inferrule*[Right=1R]{ }{\cdot \vdash 1}
  \and
  \inferrule*[Right=1L]
    {\Gamma(\cdot) \vdash A}
    {\Gamma([1]_1) \vdash A}
  \and
  \inferrule*[Right=$\bot$ L]
    { }
    {\Gamma([\bot]_s) \vdash A}
  \and
  \inferrule*[Right=\lolliR]
    {\Gamma \,_p [A]_1 \vdash B}
    {\Gamma \vdash A \multimap_p B}
  \and
  \inferrule*[Right=\lolliL]
    {\Gamma \vdash A \and \Delta([B]_s) \vdash C}
    {\Delta ([A \multimap_p B]_1 \,_p s\Gamma) \vdash C}
  \and
  \inferrule*[Right=\tensorR]
    {\Gamma \vdash A \and \Delta \vdash B}
    {\Gamma \,_p \Delta \vdash A \otimes_p B}
  \and
  \inferrule*[Right=\tensorL]
    {\Gamma([A]_s \,_p [B]_s) \vdash C}
    {\Gamma([A \otimes_p B]_s) \vdash C}
  \and
  \inferrule*[Right=\plusR1]
    {\Gamma \vdash A}
    {\Gamma \vdash A \oplus B}
  \and
  \inferrule*[Right=\plusR2]
    {\Gamma \vdash B}
    {\Gamma \vdash A \oplus B}
  \and
  \inferrule*[Right=\plusL]
    {\Gamma([A]_s) \vdash C \and \Gamma([B]_s) \vdash C}
    {\Gamma([A \oplus B]_s) \vdash C}
  \and
  \inferrule*[Right=!R]
    {\Gamma \vdash A}
    {s\Gamma \vdash \!_sA}
  \and
  \inferrule*[Right=!L]
    {\Gamma([A]_{r \cdot s}) \vdash B}
    {\Gamma([!_rA]_s) \vdash B}
  \and
  \inferrule*[Right=Contr]
    {\Gamma(\Delta \,_p \Delta') \vdash A \and \Delta \approx \Delta'}
    {\Gamma(Contr(p, \Delta, \Delta')) \vdash A}
  \and
  \inferrule*[Right=Cut]
    {\Gamma \vdash A \and \Delta([A]_s) \vdash B}
    {\Delta(s\Gamma) \vdash B}
  \and
  \inferrule*[Right=Exch]
    {\Gamma \vdash A \and \Gamma \leftrightsquigarrow \Gamma'}
    {\Gamma' \vdash A}
  \and
  \inferrule*[Right=Weak]
    {\Gamma(\cdot) \vdash A}
    {\Gamma(\Delta) \vdash A}
  \end{mathpar}
  \caption{Inference Rules of \system{}'s Logic}
\end{figure*}

\paragraph{Cut Elimination} The Cut rule is admissible in \system{}'s Logic. The
complicated part of this proof is tracking any instances of the principal
formula higher in the derivation tree. This is necessary because of the
generalized contraction rule: any principal formula can be the result of a
contraction. The main engine of cut elimination is the following theorem.

\begin{restatable}{theorem}{cuttable}
  \label{thm:cuttable}
  Given formulas $A$, $B$, context $\Gamma$, and context $\Delta$ with $n$ holes
  labeled by a set of variables $X$, two cut-free derivations $\Gamma \vdash A$
  and $\Delta\{ x \mapsto [A]_{s_x} \}_{x \in X} \vdash B$, then there is a
  cut-free derivation of $\Delta\{ x \mapsto s_x \Gamma \}_{x \in X} \vdash
  B$.
\end{restatable}

\paragraph{Semantics of the Logic} \system's logic has a similar semantics to metric spaces. Bunches have the same interpretation as in the language, and formulas interpretations are the same as their type counterparts. The semantics of derivations are described in \Cref{def:logic-derivation-semantics}
\begin{definition}
  \label{def:logic-derivation-semantics}
  Every derivation $\tau$ of $\Gamma \vdash A$ has an interpretation to a non-expansive function of type $\intrp \Gamma \rightarrow \intrp A$

  By structural induction on $\tau$.

  $\intrp {Axiom} \triangleq \lambda x .\; x $

  $\intrp {\R R } \triangleq \lambda () .\; r \in \R $

  $\intrp {1R} \triangleq \lambda () .\; 1 $

  $\intrp {1L \; \pi} \triangleq \lambda \Gamma( 1 ) .\; \intrp \pi \; \Gamma(\:()\:) $

  $\intrp {\multimap R \;\pi} \triangleq \lambda \Gamma .\; \lambda A .\; \intrp \pi \;(\Gamma , A) $

  $\intrp {\multimap L \; \pi_1 \; \pi_2} \triangleq \lambda \Delta( f , \Gamma).\; \intrp {\pi_2} \Delta( f( \intrp{\pi_1} ) )$

  $\intrp {\otimes R \; \pi_1 \; \pi_2}  \triangleq \lambda (\Gamma , \Delta) .\; (\intrp {\pi_1} \; \Gamma) , (\intrp {\pi_2} \; \Delta)$

  $\intrp {\otimes L \; \pi}  \triangleq \lambda \Gamma(a, b) .\; (\intrp {\pi} \; \Gamma( a , b ))$

  $\intrp {\oplus_iR \; \pi } \triangleq \lambda \Gamma .\; inj_i \intrp {\pi} \;\Gamma $

  $\intrp {\oplus L \; \pi_1 \; \pi_2 } \triangleq \lambda \Gamma (inj_1 \; a).\; \intrp {\pi_1} \; \Gamma( a )$

  $\intrp {\oplus L \; \pi_1 \; \pi_2 } \triangleq \lambda \Gamma (inj_2 \; b).\; \intrp {\pi_2} \; \Gamma( b )$

  $\intrp {!R \; \pi} \triangleq \intrp{\pi}$

  $\intrp {!L \; \pi} \triangleq \intrp{\pi}$

  $\intrp {Contr \; \pi } \triangleq \lambda \Gamma(\Delta).\; \intrp \pi \; \Gamma( \Delta, \Delta )$

  $\intrp {Cut \; \pi_1 \; \pi_2} \triangleq \lambda \Delta(\Gamma) .\; \intrp {\pi_2} \Delta( \intrp {\pi_1} \Gamma )$

  $\intrp {Exch \; \pi } \triangleq \lambda \Gamma'. \intrp{\pi} \Gamma $

  $\intrp {Weak \; \pi } \triangleq \lambda \Gamma(\Delta) .\; \intrp \pi \; \Gamma( \: () \: )$
\end{definition}


\begin{theorem}
  \label{thm:cut-elim-toplevel}
  The logic satisfies cut elimination: given a derivation of $\Gamma \vdash A$,
  there exists another derivation of $\Gamma \vdash A$ that does not use the cut
  rule.
\end{theorem}

\begin{proof}
  Let $\tau$ be a derivation of $\Gamma \vdash A$. Show there exists $\tau'$ which proves $\Gamma \vdash A$ that does not use a cut rule. By induction on the height of $\tau$
  \begin{itemize}
    \item One premise: recur (1)
    \item Two premise: recur (2)
    \item Cut rule
    \\ let $\pi_1$ be the derivation of $\Gamma \vdash A$ and $\pi_2$ be the derivation of $\Delta([A]_s) \vdash B$. We need to show there exists a cut-free derivation of $\Delta(s\Gamma) \vdash B$. Start by calling the IH on $\pi_1$ and $\pi_2$ to get cut-free derivations $\pi_1'$ and $\pi_2'$ their respective proofs. Now by \Cref{thm:cuttable} we can combine $\pi_1'$ and $\pi_2'$ to obtain the desired derivation.
  \end{itemize}
\end{proof}

\cuttable*

\begin{proof}
  By induction on A with inner induction on $\pi_1$ and $\pi_2$.

  \newcommand{\IH}[1]{IH$_{#1}$}

  Let \IH 1 denote induction appealing to the outer measure (size of A) and  \IH 2 denote induction appealing to the inner measure (size of $\pi_1$ and $\pi_2$).

  In each of the key cases the environment of the premise of $\pi_2$ will have exactly one hole, we name this hole $z$. If the hole is a hole that's being tracked in $\Delta$ then we transform based on the key case. If not then we use \IH2 on the premise and continue.

  First we address the non-key cases

  \begin{itemize}
    \item If there are not more holes being tracked in $\Delta$ then we are done.
    \item $\pi_1$ is a left introduction rule. Push $\pi_2$ upwards in $\pi_1$ and call \IH 2
    \item $\pi_2$ is a right introduction rule. Push $\pi_1$ upwards in $\pi_2$ and call \IH 2
    \item $\pi_2$ is a left introduction rule, but it does not introduce the principal formula of $\pi_1$. i.e. $\Gamma \vdash A$ and $\Delta(C) \vdash B$. Push $\pi_1$ upwards in $\pi_2$ and call \IH 2

    \item Suppose the last rule in $\pi_2$ is $Contr$.

      \begin{mathpar}
        \Gamma \vdash A
        \and
        \inferrule*[Right=Contr]
          {\Delta'(\Psi_1 \,_p \Psi_2) \vdash B}
          {\Delta'(Contr(p, \Psi_1, \Psi_2)) \vdash B}
      \end{mathpar}

      The situation looks like this: $\Delta\{ x \mapsto [A]_{s_x}\}_{x \in X} = \Delta'(\Psi)$, where $\Psi = Contr(p, \Psi_1 , \Psi_2)$ and $\Delta'$ is a context with only one hole, $z$. Without loss of generality we assume $z \notin X$. Some of the variables in $X$ fall in a subtree of $\Psi$, call this subset $Y$. This means that $\Psi$, $\Psi_1$, $\Psi_2$, and $\Delta'$ are of the form:
      \begin{align*}
        \Psi &= \Psi'\{ y \mapsto [A]_{L_p(a_y, b_y)}\}_{y \in Y}\\
        \Psi_1  &= \Psi_1'\{y \mapsto [A]_{a_y}\}_{y \in Y}\\
        \Psi_2  &= \Psi_2'\{y \mapsto [A]_{b_y}\}_{y \in Y}\\
        \Delta'(\cdot) &= \Delta''\begin{cases}
          x \mapsto [A]_{s_x}, x \in X \setminus Y \\
          z \mapsto \cdot \\
        \end{cases}
      \end{align*}
      Where the $a_y$'s and $b_y$'s are vectors of numbers such that $s_y = L_p(a_y, b_y)$, and $\Delta''$ is a generalization of $\Delta'$ with holes from $(X \setminus Y) \cup \{z\}$

      Now apply the following transformation to the derivation.

      \begin{mathpar}
        \inferrule*[Right=!Contr]
            {
              \inferrule*[Right=\IH 2]
                {
                  \Gamma \vdash A
                  \and
                  \Delta_1 \vdash B
                }
                {\Delta_2 \vdash B}
            }
            {\Delta_3 \vdash B}
      \end{mathpar}
      where
      \begin{align*}
        \Delta_1 &= \Delta'' \begin{cases}
          x\mapsto [A]_{s_x}, x \in X\setminus Y\\
          z \mapsto \Psi_1'\{y \mapsto [A]_{a_y}\}_{y \in Y} \,_p \Psi_2'\{y \mapsto [A]_{b_y}\}_{y \in Y}\\
        \end{cases}
        \\
        \Delta_2 &= \Delta''\begin{cases}
          x\mapsto s_x\Gamma, x \in X\setminus Y \\
          z \mapsto \Psi_1'\{y \mapsto a_y\Gamma\}_{y \in Y} \,_p \Psi_2'\{y \mapsto b_y\Gamma\}_{y \in Y}
        \end{cases}
        \\
        \Delta_3 &= \Delta'' \begin{cases}
          x\mapsto s_x\Gamma, x \in X\setminus Y \\
          z \mapsto Contr(p, \Psi_1'\{y \mapsto a_y\Gamma\}_{y \in Y} , \Psi_2')\{y \mapsto b_y\Gamma\}_{y \in Y})
        \end{cases}
      \end{align*}
    \item Last inference rule in $\pi_2$ is Weak.
      \begin{mathpar}
        \Gamma \vdash A
        \and
        \inferrule*[Right=!Weak]
          {\Delta'(\cdot) \vdash B}
          {\Delta'(\Psi)}
      \end{mathpar}

      Unifying our contexts we find that $\Delta\{x \mapsto [A]_{s_x}\}_{x \in X} = \Delta'(\Psi)$.

      Find set of variables $Y$ which correspond to the holes in a subtree of $\Psi$. Name $\Delta'$'s hole $z$. Construct a copy of $\Delta$ named $\Delta''$ with holes $(X \setminus Y )\cup \{ z \} $. Also construct a copy of $\Psi$ named $\Psi'$ with all holes in $Y$.
      \begin{align*}
        \Delta\{x \mapsto [A]_{s_x}\}_{x \in X} &= \Delta''\begin{cases}
          x \mapsto [A]_{s_x}, x \in X \setminus Y \\
          z \mapsto \Psi'\{ y\mapsto [A]_{s_y}\}_{y\in Y}\}
        \end{cases} \\
        \Psi &= \Psi'\{ y\mapsto [A]_{s_y}\}_{y\in Y}\}
      \end{align*}
      \begin{mathpar}
        \inferrule*[Right=Weak]
          {
            \inferrule*[Right=\IH 2]
              {
                \Gamma \vdash A
                \and
                \Delta''\{x \mapsto [A]_{s_x}, z \mapsto \cdot\}_{x \in X \setminus Y} \vdash B
              }
              {
                \Delta''\{x \mapsto s_x\Gamma, z \mapsto \cdot\}_{x \in X \setminus Y} \vdash B
              }
          }
          {
            \Delta''\{x \mapsto s_x\Gamma, z \mapsto \Psi'\{y \mapsto s_y\Gamma\}_{y\in Y}\}_{x\in X \setminus Y} \vdash B
          }
      \end{mathpar}
  \end{itemize}

  Now we address the key cases
  \begin{itemize}
    \item (1R, 1L)
    \begin{mathpar}
      \inferrule*[Right=1R]{ }{\cdot \vdash 1}
      \and
      \inferrule*[Right=1L]
        {\Delta(\cdot) \vdash A}
        {\Delta(1) \vdash A}
      \\\rightsquigarrow\\
      \Delta(\cdot) \vdash A
    \end{mathpar}
    \item (\lolliR , \lolliL)
      \begin{mathpar}
      \inferrule*[Right={\lolliR}]
        {\Gamma \,_p [A]_1 \vdash B}
        {\Gamma \vdash A \multimap_p B}
      \and
      \and
      \inferrule*[Right={\lolliL}]
        {\Delta \vdash A \and \Psi([B]_s) \vdash C}
        {\Psi ([A \multimap_p B]_1 \,_p s\Delta) \vdash C }

      \\\rightsquigarrow\\

      \inferrule*[Right=\IH 1]
        {
          \Delta \vdash A
          \and
          \inferrule*[Right=\IH 1]
            {\Gamma \,_p [A]_1 \vdash B \and \Psi([B]_s) \vdash C}
            {\Psi(s\Gamma \,_p [A]_{s}) \vdash C}
        }
        {\Psi(s\Gamma \,_p s\Delta) \vdash C}
      \end{mathpar}
    \item (\tensorR, \tensorL)
      \begin{mathpar}
      \inferrule*[Right=\tensorR]
        {\Gamma \vdash A \and \Delta \vdash B}
        {\Gamma \,_p \Delta \vdash A \otimes_p B}
      \and
      \inferrule*[Right=\tensorL]
        {\Psi([A]_s \,_p [B]_s) \vdash C}
        {\Psi([A \otimes_p B]_s) \vdash C}

      \\\rightsquigarrow\\

      \inferrule*[Right=\IH 1]
        {
          \Gamma \vdash A
          \and
          \inferrule*[Right=\IH 1]
            {\Delta \vdash B \and \Psi([A]_s \,_p [B]_s) \vdash C}
            {\Psi([A]_s \,_p s\Delta) \vdash C}
        }
        {\Psi(s\Gamma \,_p s\Delta) \vdash C}
      \end{mathpar}
    \item (\plusR i, \plusL)
      \begin{mathpar}
        \inferrule*[Right=\plusR i]
          {\Gamma \vdash A_i}
          {\Gamma \vdash A_1 \oplus A_2}
        \and
        \inferrule*[Right=\plusL]
          {\Delta([A_1]_s) \vdash B \and \Delta([A_2]_s) \vdash B}
          {\Delta([A_1 \oplus A_2]_s) \vdash B}
        \\\rightsquigarrow\\
        \inferrule*[Right=\IH 1]
          {
            \Gamma \vdash A_i
            \and
            \Delta([A_i]_s) \vdash B
          }
          {\Delta(s\Gamma) \vdash B}
      \end{mathpar}
    \item (!R, \!L)
      \begin{mathpar}
        \inferrule*[Right=!R]
          {\Gamma \vdash A}
          {s\Gamma \vdash \!_sA}
        \and
        \inferrule*[Right=!L]
          {\Delta([A]_{s \cdot r}) \vdash B}
          {\Delta([!_sA]_r) \vdash B}
        \\\rightsquigarrow\\
        \inferrule*[Right=\IH 1]
          {\Gamma \vdash A \and \Delta([A]_{s \cdot r}) \vdash B}
          {\Delta(sr\Gamma \vdash B}
      \end{mathpar}
  \end{itemize}
\end{proof}

\fi

\end{document}

%% file: sect_examplesESOP.tex
We now look at examples of programs that illustrate the use of \lp metrics.

\paragraph{Currying and Uncurrying}

%
Let us illustrate the use of higher-order functions with combinators for
currying and uncurrying.
\begin{align*}
  curry &: ((\tau\otimes_p \sigma) \multimap_p \rho) \multimap (\tau\multimap_p \sigma \multimap_p \rho)\\
  curry &\;f \; x \; y = f (x, y)\\
  uncurry &: (\tau\multimap_p \sigma \multimap_p \rho) \multimap ((\tau\otimes_p \sigma) \multimap_p \rho).\\
  uncurry &\;f \; z = \LetPairIn{x}{y}{z}{f\; x\; y}
\end{align*}
Note that the indices on $\otimes$ and $\multimap$ need to be the same. The
reason can be traced back to the $\multimap\mathrm{E}$ rule
(cf. \Cref{fig:BunchedFuzzTypingRules}), which uses the $,_p$ connective to
eliminate $\multimap_p$ \ifappendix (cf. \Cref{fig:curry-deriv} in the Appendix
  for a detailed derivation).  \else (cf. the currying and uncurrying derivation
  in the appendix of the full paper for a detailed derivation).  \fi If the
indices do not agree, currying is not possible; in other words, we cannot in
general soundly curry a function of type $\tau\otimes_p \sigma \multimap_q \rho$
to obtain something of type $\tau \multimap_p \sigma \multimap_q \rho$.
However, if $q \leq p$, note that it would be possible to soundly view
$\tau \otimes_q \sigma$ as a subtype of $\tau \otimes_p \sigma$, thanks to
\Cref{prop:relationTensors}.  In this case, we could then convert from
$\tau \otimes_p \sigma \multimap_q \rho$ to
$\tau \otimes_q \sigma \multimap_q \rho$ (note the variance), and then curry to
obtain a function of type $\tau \multimap_q \sigma \multimap_q \rho$.

\aaa{We should include an example demonstrating why this is not sound in
  general.}

\paragraph{Precise sensitivity for functions with multiple arguments}

Another useful feature of \system{} is that its contraction rule allows us to
split sensitivities more accurately than if we used the contraction rule that is
derivable in the original Fuzz.  Concretely, suppose that we have a program
$\lambda p. \LetPairIn{x}{y}{p} f(x,y) + g(x,y)$, where $f$ and $g$ have types
$f : (!_2\R) \otimes_2 \R \multimap \R$ and
$g : \R \otimes_2 (!_2\R) \multimap \R$, and where we have elided the wrapping
and unwrapping of $!$ types, for simplicity.

Let us sketch how this program is typed in \system{}. Addition belongs to
$\R \otimes_1 \R \multimap \R$, so by \Cref{prop:relationTensors} it can also be
given the type $!_{\sqrt {2}} (\R \otimes_2 \R) \multimap \R$. Thus, we can
build the following derivation for the body of the program:
\begin{mathpar}
\mprset { fraction ={===}}
  \inferrule*[Left=Contr]{
\Gamma \vdash  f(x_1,y_1) + g(x_2,y_2) : \R }
 {[x :\R ]_{\sqrt {10}} \; ,_2 [y :\R ]_{\sqrt {10}} \vdash  f(x,y) + g(x,y) : \R }
\end{mathpar}
where
$\Gamma= ([x_1 :\R ]_{2\sqrt {2}} ,_2 [y_1 :\R ]_{\sqrt {2}}) ,_2 ([x_2 :\R
]_{\sqrt {2}} ,_2 [y_2 :\R ]_{2\sqrt {2}})$, and where we used contraction twice
to merge the $x$s and $y$s. Note that
$||(2\sqrt {2}, \sqrt {2})||_2=\sqrt {8+2}=\sqrt {10}$, which is why the final
sensitivities have this form.
By contrast, consider how we might attempt to type this program directly in the
original Fuzz. Let us assume that we are working in an extension of Fuzz with
types for expressing the domains of $f$ and $g$, similarly to the $L^2$ vector
types of Duet~\cite{DBLP:journals/pacmpl/NearDASGWSZSSS19}.  Moreover, let us
assume that we have coercion functions that allow us to cast from
$(!_2\R) \otimes_2 (!_2\R)$ to $(!_2\R) \otimes_2 \R$ and
$\R \otimes_2 (!_2 \R)$.  If we have a pair
$p : !_2 ((!_2 \R) \otimes_2 (!_2\R))$, we can split its sensitivity to call $f$
and $g$ and then combine their results with addition.  However, this type is
equivalent to $!_4 (\R \otimes_2 \R)$, which means that the program was given a
worse sensitivity (since $\sqrt {10}< 4$).
Of course, it would also have been possible to extend Fuzz with a series of
primitives that implement precisely the management of sensitivities performed by
bunches.  However, here this low-level reasoning is handled directly by the type
system.

\aaa{This last point is a bit weak, since we do not have an implementation of
  the type checker.}

\paragraph{Programming with matrices}

The Duet language~\cite{DBLP:journals/pacmpl/NearDASGWSZSSS19} provides several
matrix types with the $L^1$, $L^2$, or $L^\infty$ metrics, along with primitive
functions for manipulating them.  In \system{}, these types can be defined
directly as follows: $\mathbb{M}_p[m,n] = \otimes_1^m \otimes_p^n \R$. Following
Duet, we use the $L^1$ distance to combine the rows and the $L^p$ distance to
combine the columns.
One advantage of having types for matrices defined in terms of more basic
constructs is that we can program functions for manipulating them directly,
without resorting to separate primitives.  For example, we can define the
following terms in the language:
\begin{align*}
addrow &: \mathbb{M}_{p} [1,n] \otimes_1 \mathbb{M}_{p} [m,n]
\multimap \mathbb{M}_{p} [m+1,n]  \\
addcolumn &: \mathbb{M}_{1} [1,m] \otimes_1 \mathbb{M}_{1} [m,n]
\multimap \mathbb{M}_{1} [m,n+1] \\
addition &: \mathbb{M}_{1} [m,n] \otimes_1 \mathbb{M}_{1} [m,n] \multimap \mathbb{M}_{1} [m,n].
\end{align*}
The first program, $addrow$, appends a vector, represented as a $1 \times n$
matrix, to the first row of a $m\times n$ matrix. The second program,
$addcolumn$, is similar, but appends the vector as a column rather than a row.
Because of that, it is restricted to $L^1$ matrices.
\aaa{Can we generalize addcolumn to $L^p$ matrices? It feels like it should be
  possible, at least if we use $\otimes_p$ to separate the two arguments.}
Finally, the last program, $addition$, adds the elements of two matrices
pointwise.


\paragraph{Vector addition over sets}
Let us now show an example of a Fuzz term for which using $L^p$ metrics allows
to obtain a finer sensitivity analysis.
We consider sets of vectors in $\R^d$ and the function $\mathit{vectorSum}$
which, given such a set, returns the vectorial sum of its elements.  In Fuzz,
this function can be defined via a summation primitive
$\mathit{sum} : {!_\infty (!_\infty \tau \multimap \R)} \multimap \Set \tau
\multimap \R,$ which adds up the results of applying a function to each element
of a set~\cite{DBLP:conf/icfp/ReedP10}.  The definition is:
\begin{align*}
  \mathit{vectorSum} & : {!_d \Set(\otimes_1^d \R)} \multimap_1 \otimes_1^d \R \\
  \mathit{vectorSum}\;s & = (\mathit{sum}\;\pi_1\;s,\ldots,\mathit{sum}\;\pi_d\;s).
\end{align*}
Here, $\pi_i : \otimes_1^d \R \multimap \R$ denotes the $i$-th projection, which
can be defined by destructing a product.  Set types in Fuzz are equipped with
the Hamming metric \cite{DBLP:conf/icfp/ReedP10}, where the distance between two
sets is the number of elements by which they differ.  Note that, to ensure that
$\mathit{sum}$ has bounded sensitivity, we need to clip the results of its
function argument to the interval $[-1,1]$.  Fuzz infers a sensitivity of $d$
for this function because its argument is used with sensitivity 1 in each
component of the tuple.
In \system{}, we can define the same function as above, but we also have the
option of using a different $L^p$ distance to define
$\mathit{vectorSum}$, which leads to the type
$!_{d^{1/p}} \Set(\otimes_p^d \R) \multimap \otimes_p^d \R$, with a sensitivity
of $d^{1/p}$.
For the sake of readability, we'll show how this term is typed in the case
$d=2$.  By typing each term  $(sum \; \pi_i \;z_i)$ and applying
$(\otimes I)$ we get:
$$  [z_1: \Set (\R \otimes_p \R) ]_1 \; ,_p  [z_2: \Set (\R \otimes_p \R)
]_1\vdash (sum \; \pi_1 \;z_1,sum \; \pi_2 \;z_2): \R \otimes_p \R. $$
By
applying contraction we get:
$ [z : \Set (\R \otimes_p \R) ]_{2^{1/p}}  \vdash (sum \; \pi_1
\;z,sum \; \pi_2 \;z): \R \otimes_p \R. $
The claimed type is finally obtained
by $(! E)$ and $(\multimap I)$.

\paragraph{Computing distances}

Suppose that the type $X$ denotes a proper metric space (that is, where the
triangle inequality holds).  Then we can incorporate its distance function in
\system{} with the type $X \otimes_1 X \multimap \R$.
Indeed, let $x$, $x'$, $y$ and $y'$ be arbitrary elements of $X$.  Then
\begin{align*}
d_X(x,y) - d_X(x',y')
& \leq d_X(x,x') + d_X(x',y') + d_X(y',y) - d_X(x',y') \\
& = d_X(x,x') + d_X(y,y') = d_1((x,y),(x',y')).
\end{align*}
By symmetry, we also know that $d_X(x',y') - d_X(x,y) \leq d_1((x,y),(x',y'))$.
Combined, these two facts show
\begin{align*}
  d_{\R}(d_X(x,y), d_X(x',y'))
  & = |d_X(x,y) - d_X(x',y')| \leq d_1((x,y),(x',y')),
\end{align*}
which proves that $d_X$ is indeed a non-expansive function.

\paragraph{Calibrating noise to $L^p$\/distance}
Hardt and Talwar~\cite{HardtT10} have proposed a generalization of the Laplace
mechanism, called the $K$-norm mechanism, to create a differentially private
variant of a database query $f: \db\to \mathbb{R}^d$.  The difference is that
the amount of noise added is calibrated to the sensitivity of $f$ measured with
the $K$ norm, as opposed to the $L^1$ distance used in the original Laplace
mechanism.  When $K$ corresponds to the $L^p$ norm, we will call this the
$L^p$-mechanism, following Awan and Slavkovich~\cite{osti_10183971}.

\begin{definition}
  Given $f:\db\to \mathbb{R}^d$ with $L^p$ sensitivity $s$ and $\epsilon>0$, the
  $L^p$-mechanism is a mechanism that, given a database $D\in\db$, returns a
  probability distribution over $y\in \mathbb{R}^d$ with density given by:
\[
\frac{\exp(\frac{-\epsilon||f(D)-y||_p}{2s})}
{\int \exp(\frac{-\epsilon||f(D)-y||_p}{2s})d y}
\]
\end{definition}
This mechanism returns with high probability (which depends on $\epsilon$ and on
the sensitivity $s$) a vector $y\in\mathbb{R}^d$ which is close to $f(D)$ in
$L^p$ distance.  Such a mechanism can be easily integrated in \system{} through
a primitive:
$$
{\tt LpMech} :  {!_{\infty}(!_s {\tt dB}\multimap \otimes_p^d\mathbb{R})} \multimap {!_\epsilon {\tt dB}} \multimap  \bigcirc_P (\otimes^d_p\mathbb{R})
$$
(Strictly speaking, we would need some discretized version of the above
distribution to incorporate the mechanism in \system{}, but we'll ignore this
issue in what follows.)
The fact that ${\tt LpMech}$ satisfies $\epsilon$-differential privacy follows
from the fact that this mechanism is an instance of the \emph{exponential
  mechanism}~\cite{McSherryT07}, a basic building block of differential
privacy. It is based on a scoring function assigning a score to every pair
consisting of a database and a potential output, and it attempts to return an
output with approximately maximal score, given the input database. As shown by
Gaboardi et al.~\cite{DBLP:conf/popl/GaboardiHHNP13}, the exponential mechanism
can be added as a primitive to Fuzz with type:
$$
{\tt expmech} : {!_\infty \Set(\mathcal{O})} \multimap {!_\infty (!_\infty
  \mathcal{O} \multimap !_s {\tt dB}\multimap \mathbb{R})} \multimap
!_\epsilon{\tt dB} \multimap \bigcirc_P \mathcal{O},
$$
where $\mathcal{O}$ is the type of outputs.  The function ${\tt LpMech}$ is an
instance of the exponential mechanism where $\mathcal{O}$ is
$\otimes_p^d\mathbb{R}$ and the score is $\lambda y \lambda D. ||f(D) - y||_p$.


To define the $L^p$ mechanism with this recipe, we need to reason about the
sensitivity of this scoring function.  In Fuzz, this would not be possible,
since the language does not support reasoning about the sensitivity of $f$
measured in the $L^p$ distance.  In \system{}, however, this can be done
easily. 
Below, we will see an example (Gradient descent) of how the $L^p$ mechanism can
lead to a finer privacy guarantee.  \pb{Modified this last sentence as finally
  we are not using the $L^p$ mechanism for k-means but only for gradient
  descent.}

\paragraph{Gradient descent} Let us now give an example where we use
the $L^p$ mechanism. An example of differentially private gradient descent example
with linear model in Fuzz was given in
\cite{DBLP:journals/pacmpl/Winograd-CortHR17} (see Sect. 4.1, 4.2 and
Fig. 6 p. 16, Fig. 8 p.19).
 This algorithm proceeds by iteration. Actually it was given for an extended language called
Adaptative Fuzz, but the code already gives an algorithm in (plain)
Fuzz. We refer the reader to this reference for the description of all
functions, and here we will only describe how one can adapt the
algorithm to \system{}.

Given a set of $n$ records $x_i \in \R^d$, each with a \textit{label}
$y_i \in \R$, the goal is to find a parameter vector $\theta \in \R^d$ that
minimizes the difference between the labels and their \textit{estimates}, where
the estimate of a label $y_i$ is the inner product $\langle
x_i,\theta\rangle$. That is, the goal is to minimize the loss function
$L(\theta ,(x,y))= \frac{1}{n}\cdot \Sigma_{i=1}^{n}(\langle
x_i,\theta\rangle-y_i)^2.$ The algorithm starts with an initial parameter vector
$(0,\dots,0)$ and it iteratively produces successive $\theta$ vectors until a
termination condition is reached.

The Fuzz program uses the data-type $bag \; \tau$ representing bags or multisets
over $\tau$. A $bagmap$ primitive is given for it. The type $I$ is the unit
interval $[0,1]$.
The main function is called $updateParameter$ and updates one component of the
model $\theta$; it is computed in the following way:
\begin{itemize}
\item with the function $calcGrad : \db \to \R$, compute a component
  $(\nabla L( \theta,(x,y)))_j$ of the $\R^d$ vector $\nabla L( \theta,(x,y))$
  \footnote{Actually $calcGrad$ computes $(\nabla L( \theta,(x,y)))_j$ up to a
    multiplicative constant, 2/n, which is mutliplied afterwards in the
    $updateParameter$ function. }.
\item then Laplacian noise is postcomposed with $calcGrad$ in the
  $updateParameter$ function. This uses a privacy budget of $2\epsilon$. It has
  to be done for each one of the $d$ components of $\nabla L( \theta,(x,y))$,
  thus on the whole, for one step, a privacy budget of $2d\epsilon$.
\item The iterative procedure of gradient descent is given by the function
  $gradient$ in Fig. 8 p. 19 of
  \cite{DBLP:journals/pacmpl/Winograd-CortHR17}. We forget here about the
  adaptative aspect and just consider iteration with a given number $n$ of
  steps. In this case by applying $n$ times $updateParameter$ one gets a privacy
  budget of $2d n\epsilon$.
\end{itemize}

We modify the program as follows to check it in \system{} and use the
$L^p$-mechanism.
Instead of computing over $\R$ we want to compute over $ \otimes_p^d \R$ for a
given $p\geq 1$, so $\R^d$ equipped with \lp distance. The records $x_i$ are in
$ \otimes_p^d I$ and the labels $y_i$ in $I$. The database type is
${\tt dB}=bag \; (I \otimes_p(\otimes_p^d I)) $. The distance between two bags in
${\tt dB}$ is the number of elements by which they differ.

We assume a primitive $ bagVectorSum$ with type
$!_{d^{1/p}} bag \; ( \otimes_p^d I) \multimap \otimes_p^d \R$ (it could be
defined as the $\mathit{vectorSum}$ defined above for sets, using a $sum$
primitive for bags). Given a bag $m$, $(bagVectorSum \; m)$ returns the
vectorial sum of all elements of $m$. We can check that the sensitivity of
$ bagVectorSum$ is indeed $d^{1/p}$ because given two bags $m$ and $m'$ that are
at distance 1, if we denote by $u$ the vector by which they differ, we have:
\begin{align*}
d_{(\otimes_p^d \R) }(bagVectorSum (m),bagVectorSum (m'))   &= ||u
||_p \\
  & \leq  (\Sigma_{j=1}^d 1)^{1/p}  = d^{1/p}
\end{align*}
%
By adapting the $calcGrad$ Fuzz term of
\cite{DBLP:journals/pacmpl/Winograd-CortHR17} using $ bagVectorSum$ we obtain a
term $VectcalcGrad$ with the \system{} type
$!_{\infty} \otimes_p^d \R \multimap !_{d^{1/p}} \db \multimap \otimes_p^d \R$.
Given a vector $\theta$ and a database $(y,x)$, $ VectcalcGrad$ computes the
updated vector $\theta'$. Finally we define the term $updateVector$ by adding
noise to $ VectcalcGrad$ using the the $L^p$-mechanism. Recall the type of
${\tt LpMech}$:
$ !_{\infty}(!_s \db\multimap \otimes_p^d\mathbb{R})\multimap !_\epsilon\db
\multimap \bigcirc_P (\otimes^d_p\mathbb{R}).  $ We define $ updateVector$ and
obtain its type as follows:
$$ updateVector= \lambda \theta. ({\tt LpMech} \; (VectcalcGrad \;
\theta)):  \; !_{\infty} \otimes_p^d \R \multimap !_\epsilon \db
\multimap  \bigcirc_P ( \otimes_p^d\R) $$
By iterating $ updateVector$ $n$ times one obtains a privacy budget of $n
\epsilon$.

\pb{I am not quite sure about how the privacy budget of this Bunched
  Fuzz/LpMech version compares to the privacy budget of the original
  Fuzz/Laplace mech version. }


%% file: sect_implementation.tex
To experiment with the \system{} design, we implemented a prototype for a
fragment of the system based on DFuzz~\cite{DBLP:conf/popl/GaboardiHHNP13,
DBLP:conf/ifl/AmorimGAH14}.\footnote{%
\url{https://github.com/junewunder/bunched-fuzz}} The type-checker generates a
set of numeric constraints that serve as verification conditions to guarantee a
valid typing.
The implementation required adapting some of the current rules to an algorithmic
formulation (found in \ifappendix \Cref{fig:alg-rules}). \else the full
version). \fi In addition to the modifications introduced in the DFuzz type
checker compared to its original version~\cite{DBLP:conf/popl/GaboardiHHNP13,
DBLP:conf/ifl/AmorimGAH14}, we also made the following changes and
simplifications:
\begin{itemize}
\item We did not include explicit contraction and weakening rules. Instead, the
rules are combined with those for checking other syntactic constructs. To do
away with an explicit contraction rule, in rules that have multiple antecedents,
such as the \tensorI{} rule, we used the $Contr$ operator to combine the
antecedents' environments, rather than using the $p$-concatenation operator for
bunches.
\item We did not include the rules for checking probabilistic programs with the
Hellinger distance.
\item Bound variables are always added at the top of the current environment, as
in the \lolliI{} rule of the original rules; it is not possible to introduce new
variables arbitrarily deep in the environment.
\end{itemize}
While, strictly speaking, the resulting system is incomplete with respect to the
rules presented here, it is powerful enough to check \ifappendix the K-means
example of \Cref{ex:k-means}. \else an implementation of K-means that
generalizes a previous version implemented for
Fuzz~\cite{DBLP:conf/icfp/ReedP10}.\fi{} On the other hand, because our
implementation is based on the one of DFuzz, which features dependent types, we
allow functions that are polymorphic on types, sizes and $p$ parameters, which
allows us to infer sensitivity information that depends on run-time sizes.

%% file: arxiv.bbl

\begin{thebibliography}{27}


\ifx \showCODEN    \undefined \def \showCODEN     #1{\unskip}     \fi
\ifx \showDOI      \undefined \def \showDOI       #1{#1}\fi
\ifx \showISBNx    \undefined \def \showISBNx     #1{\unskip}     \fi
\ifx \showISBNxiii \undefined \def \showISBNxiii  #1{\unskip}     \fi
\ifx \showISSN     \undefined \def \showISSN      #1{\unskip}     \fi
\ifx \showLCCN     \undefined \def \showLCCN      #1{\unskip}     \fi
\ifx \shownote     \undefined \def \shownote      #1{#1}          \fi
\ifx \showarticletitle \undefined \def \showarticletitle #1{#1}   \fi
\ifx \showURL      \undefined \def \showURL       {\relax}        \fi
\providecommand\bibfield[2]{#2}
\providecommand\bibinfo[2]{#2}
\providecommand\natexlab[1]{#1}
\providecommand\showeprint[2][]{arXiv:#2}

\bibitem[Awan and Slavkovic(2020)]%
        {osti_10183971}
\bibfield{author}{\bibinfo{person}{Jordan Awan} {and}
  \bibinfo{person}{Aleksandra Slavkovic}.} \bibinfo{year}{2020}\natexlab{}.
\newblock \showarticletitle{Structure and Sensitivity in Differential Privacy:
  Comparing K-Norm Mechanisms}.
\newblock \bibinfo{journal}{\emph{J. Amer. Statist. Assoc.}}
  (\bibinfo{year}{2020}).
\newblock
\urldef\tempurl%
\url{https://doi.org/10.1080/01621459.2020.1773831}
\showDOI{\tempurl}


\bibitem[{Azevedo de Amorim} et~al\mbox{.}(2014)]%
        {DBLP:conf/ifl/AmorimGAH14}
\bibfield{author}{\bibinfo{person}{Arthur {Azevedo de Amorim}},
  \bibinfo{person}{Marco Gaboardi}, \bibinfo{person}{Emilio Jes{\'{u}}s~Gallego
  Arias}, {and} \bibinfo{person}{Justin Hsu}.} \bibinfo{year}{2014}\natexlab{}.
\newblock \showarticletitle{Really Natural Linear Indexed Type Checking}. In
  \bibinfo{booktitle}{\emph{Proceedings of the 26th 2014 International
  Symposium on Implementation and Application of Functional Languages, {IFL}
  '14, Boston, MA, USA, October 1-3, 2014}},
  \bibfield{editor}{\bibinfo{person}{Sam Tobin{-}Hochstadt}} (Ed.).
  \bibinfo{publisher}{{ACM}}, \bibinfo{pages}{5:1--5:12}.
\newblock
\urldef\tempurl%
\url{https://doi.org/10.1145/2746325.2746335}
\showDOI{\tempurl}


\bibitem[{Azevedo de Amorim} et~al\mbox{.}(2019)]%
        {DBLP:conf/lics/AmorimGHK19}
\bibfield{author}{\bibinfo{person}{Arthur {Azevedo de Amorim}},
  \bibinfo{person}{Marco Gaboardi}, \bibinfo{person}{Justin Hsu}, {and}
  \bibinfo{person}{Shin{-}ya Katsumata}.} \bibinfo{year}{2019}\natexlab{}.
\newblock \showarticletitle{Probabilistic Relational Reasoning via Metrics}. In
  \bibinfo{booktitle}{\emph{34th Annual {ACM/IEEE} Symposium on Logic in
  Computer Science, {LICS} 2019, Vancouver, BC, Canada, June 24-27, 2019}}.
  \bibinfo{publisher}{{IEEE}}, \bibinfo{pages}{1--19}.
\newblock
\urldef\tempurl%
\url{https://doi.org/10.1109/LICS.2019.8785715}
\showDOI{\tempurl}


\bibitem[{Azevedo de Amorim} et~al\mbox{.}(2017)]%
        {DBLP:conf/popl/AmorimGHKC17}
\bibfield{author}{\bibinfo{person}{Arthur {Azevedo de Amorim}},
  \bibinfo{person}{Marco Gaboardi}, \bibinfo{person}{Justin Hsu},
  \bibinfo{person}{Shin{-}ya Katsumata}, {and} \bibinfo{person}{Ikram
  Cherigui}.} \bibinfo{year}{2017}\natexlab{}.
\newblock \showarticletitle{A semantic account of metric preservation}. In
  \bibinfo{booktitle}{\emph{{POPL} 2017}}. \bibinfo{publisher}{{ACM}}.
\newblock
\urldef\tempurl%
\url{http://dl.acm.org/citation.cfm?id=3009890}
\showURL{%
\tempurl}


\bibitem[Bao et~al\mbox{.}(2022)]%
        {10.1145/3498719}
\bibfield{author}{\bibinfo{person}{Jialu Bao}, \bibinfo{person}{Marco
  Gaboardi}, \bibinfo{person}{Justin Hsu}, {and} \bibinfo{person}{Joseph
  Tassarotti}.} \bibinfo{year}{2022}\natexlab{}.
\newblock \showarticletitle{A Separation Logic for Negative Dependence}.
\newblock \bibinfo{journal}{\emph{Proc. ACM Program. Lang.}}
  \bibinfo{volume}{6}, \bibinfo{number}{POPL}, Article \bibinfo{articleno}{57}
  (\bibinfo{date}{jan} \bibinfo{year}{2022}), \bibinfo{numpages}{29}~pages.
\newblock
\urldef\tempurl%
\url{https://doi.org/10.1145/3498719}
\showDOI{\tempurl}


\bibitem[Barthe et~al\mbox{.}(2016)]%
        {DBLP:conf/ccs/BartheFGAGHS16}
\bibfield{author}{\bibinfo{person}{Gilles Barthe}, \bibinfo{person}{Gian~Pietro
  Farina}, \bibinfo{person}{Marco Gaboardi}, \bibinfo{person}{Emilio
  Jes{\'{u}}s~Gallego Arias}, \bibinfo{person}{Andy Gordon},
  \bibinfo{person}{Justin Hsu}, {and} \bibinfo{person}{Pierre{-}Yves Strub}.}
  \bibinfo{year}{2016}\natexlab{}.
\newblock \showarticletitle{Differentially Private Bayesian Programming}. In
  \bibinfo{booktitle}{\emph{Proceedings of the 2016 {ACM} {SIGSAC} Conference
  on Computer and Communications Security, Vienna, Austria, October 24-28,
  2016}}, \bibfield{editor}{\bibinfo{person}{Edgar~R. Weippl},
  \bibinfo{person}{Stefan Katzenbeisser}, \bibinfo{person}{Christopher
  Kruegel}, \bibinfo{person}{Andrew~C. Myers}, {and} \bibinfo{person}{Shai
  Halevi}} (Eds.). \bibinfo{publisher}{{ACM}}, \bibinfo{pages}{68--79}.
\newblock
\urldef\tempurl%
\url{https://doi.org/10.1145/2976749.2978371}
\showDOI{\tempurl}


\bibitem[Barthe and Olmedo(2013)]%
        {DBLP:conf/icalp/BartheO13}
\bibfield{author}{\bibinfo{person}{Gilles Barthe} {and}
  \bibinfo{person}{Federico Olmedo}.} \bibinfo{year}{2013}\natexlab{}.
\newblock \showarticletitle{Beyond Differential Privacy: Composition Theorems
  and Relational Logic for f-divergences between Probabilistic Programs}. In
  \bibinfo{booktitle}{\emph{Automata, Languages, and Programming - 40th
  International Colloquium, {ICALP} 2013, Riga, Latvia, July 8-12, 2013,
  Proceedings, Part {II}}} \emph{(\bibinfo{series}{Lecture Notes in Computer
  Science}, Vol.~\bibinfo{volume}{7966})},
  \bibfield{editor}{\bibinfo{person}{Fedor~V. Fomin}, \bibinfo{person}{Rusins
  Freivalds}, \bibinfo{person}{Marta~Z. Kwiatkowska}, {and}
  \bibinfo{person}{David Peleg}} (Eds.). \bibinfo{publisher}{Springer},
  \bibinfo{pages}{49--60}.
\newblock
\urldef\tempurl%
\url{https://doi.org/10.1007/978-3-642-39212-2\_8}
\showDOI{\tempurl}


\bibitem[Bousquet and Elisseeff(2002)]%
        {DBLP:journals/jmlr/BousquetE02}
\bibfield{author}{\bibinfo{person}{Olivier Bousquet} {and}
  \bibinfo{person}{Andr{\'{e}} Elisseeff}.} \bibinfo{year}{2002}\natexlab{}.
\newblock \showarticletitle{Stability and Generalization}.
\newblock \bibinfo{journal}{\emph{J. Mach. Learn. Res.}}  \bibinfo{volume}{2}
  (\bibinfo{year}{2002}), \bibinfo{pages}{499--526}.
\newblock
\urldef\tempurl%
\url{http://jmlr.org/papers/v2/bousquet02a.html}
\showURL{%
\tempurl}


\bibitem[Boyd and Vandenberghe(2004)]%
        {coBook04}
\bibfield{author}{\bibinfo{person}{Stephen Boyd} {and} \bibinfo{person}{Lieven
  Vandenberghe}.} \bibinfo{year}{2004}\natexlab{}.
\newblock \bibinfo{booktitle}{\emph{Convex Optimization}}.
\newblock \bibinfo{publisher}{{Cambridge University Press}}.
\newblock
\showISBNx{0521833787}


\bibitem[Chaudhuri et~al\mbox{.}(2011)]%
        {DBLP:conf/sigsoft/ChaudhuriGLN11}
\bibfield{author}{\bibinfo{person}{Swarat Chaudhuri}, \bibinfo{person}{Sumit
  Gulwani}, \bibinfo{person}{Roberto Lublinerman}, {and} \bibinfo{person}{Sara
  NavidPour}.} \bibinfo{year}{2011}\natexlab{}.
\newblock \showarticletitle{Proving programs robust}. In
  \bibinfo{booktitle}{\emph{SIGSOFT/FSE'11 19th {ACM} {SIGSOFT} Symposium on
  the Foundations of Software Engineering {(FSE-19)} and ESEC'11: 13th European
  Software Engineering Conference (ESEC-13), Szeged, Hungary, September 5-9,
  2011}}, \bibfield{editor}{\bibinfo{person}{Tibor Gyim{\'{o}}thy} {and}
  \bibinfo{person}{Andreas Zeller}} (Eds.). \bibinfo{publisher}{{ACM}},
  \bibinfo{pages}{102--112}.
\newblock
\urldef\tempurl%
\url{https://doi.org/10.1145/2025113.2025131}
\showDOI{\tempurl}


\bibitem[Csiszár and Shields(2004)]%
        {CIT-004}
\bibfield{author}{\bibinfo{person}{I. Csiszár} {and} \bibinfo{person}{P.C.
  Shields}.} \bibinfo{year}{2004}\natexlab{}.
\newblock \showarticletitle{Information Theory and Statistics: A Tutorial}.
\newblock \bibinfo{journal}{\emph{Foundations and Trends® in Communications
  and Information Theory}} \bibinfo{volume}{1}, \bibinfo{number}{4}
  (\bibinfo{year}{2004}), \bibinfo{pages}{417--528}.
\newblock
\showISSN{1567-2190}
\urldef\tempurl%
\url{https://doi.org/10.1561/0100000004}
\showDOI{\tempurl}


\bibitem[Dwork et~al\mbox{.}(2006)]%
        {DBLP:conf/tcc/DworkMNS06}
\bibfield{author}{\bibinfo{person}{Cynthia Dwork}, \bibinfo{person}{Frank
  McSherry}, \bibinfo{person}{Kobbi Nissim}, {and} \bibinfo{person}{Adam~D.
  Smith}.} \bibinfo{year}{2006}\natexlab{}.
\newblock \showarticletitle{Calibrating Noise to Sensitivity in Private Data
  Analysis}. In \bibinfo{booktitle}{\emph{Theory of Cryptography, Third Theory
  of Cryptography Conference, {TCC} 2006, New York, NY, USA, March 4-7, 2006,
  Proceedings}} \emph{(\bibinfo{series}{Lecture Notes in Computer Science},
  Vol.~\bibinfo{volume}{3876})}, \bibfield{editor}{\bibinfo{person}{Shai
  Halevi} {and} \bibinfo{person}{Tal Rabin}} (Eds.).
  \bibinfo{publisher}{Springer}, \bibinfo{pages}{265--284}.
\newblock
\showISBNx{3-540-32731-2}
\urldef\tempurl%
\url{https://doi.org/10.1007/11681878\_14}
\showDOI{\tempurl}


\bibitem[Dwork and Roth(2014)]%
        {DworkR14}
\bibfield{author}{\bibinfo{person}{Cynthia Dwork} {and} \bibinfo{person}{Aaron
  Roth}.} \bibinfo{year}{2014}\natexlab{}.
\newblock \showarticletitle{The Algorithmic Foundations of Differential
  Privacy}.
\newblock \bibinfo{journal}{\emph{Found. Trends Theor. Comput. Sci.}}
  \bibinfo{volume}{9}, \bibinfo{number}{3-4} (\bibinfo{year}{2014}),
  \bibinfo{pages}{211--407}.
\newblock
\urldef\tempurl%
\url{https://doi.org/10.1561/0400000042}
\showDOI{\tempurl}


\bibitem[Gaboardi et~al\mbox{.}(2013)]%
        {DBLP:conf/popl/GaboardiHHNP13}
\bibfield{author}{\bibinfo{person}{Marco Gaboardi}, \bibinfo{person}{Andreas
  Haeberlen}, \bibinfo{person}{Justin Hsu}, \bibinfo{person}{Arjun Narayan},
  {and} \bibinfo{person}{Benjamin~C. Pierce}.} \bibinfo{year}{2013}\natexlab{}.
\newblock \showarticletitle{Linear dependent types for differential privacy}.
  In \bibinfo{booktitle}{\emph{{POPL} '13}}. \bibinfo{publisher}{{ACM}}.
\newblock
\urldef\tempurl%
\url{https://doi.org/10.1145/2429069.2429113}
\showDOI{\tempurl}


\bibitem[Girard(1987)]%
        {DBLP:journals/tcs/Girard87}
\bibfield{author}{\bibinfo{person}{Jean{-}Yves Girard}.}
  \bibinfo{year}{1987}\natexlab{}.
\newblock \showarticletitle{Linear Logic}.
\newblock \bibinfo{journal}{\emph{Theor. Comput. Sci.}}  \bibinfo{volume}{50}
  (\bibinfo{year}{1987}), \bibinfo{pages}{1--102}.
\newblock
\urldef\tempurl%
\url{https://doi.org/10.1016/0304-3975(87)90045-4}
\showDOI{\tempurl}


\bibitem[Gonin and Money(1989)]%
        {10.5555/64130}
\bibfield{author}{\bibinfo{person}{Ren\'{e} Gonin} {and}
  \bibinfo{person}{Arthur~H. Money}.} \bibinfo{year}{1989}\natexlab{}.
\newblock \bibinfo{booktitle}{\emph{Nonlinear Lp-Norm Estimation}}.
\newblock \bibinfo{publisher}{Marcel Dekker, Inc.}, \bibinfo{address}{USA}.
\newblock
\showISBNx{0824781252}


\bibitem[Haeberlen et~al\mbox{.}(2011)]%
        {DBLP:conf/uss/HaeberlenPN11}
\bibfield{author}{\bibinfo{person}{Andreas Haeberlen},
  \bibinfo{person}{Benjamin~C. Pierce}, {and} \bibinfo{person}{Arjun Narayan}.}
  \bibinfo{year}{2011}\natexlab{}.
\newblock \showarticletitle{Differential Privacy Under Fire}. In
  \bibinfo{booktitle}{\emph{20th {USENIX} Security Symposium, San Francisco,
  CA, USA, August 8-12, 2011, Proceedings}}. \bibinfo{publisher}{{USENIX}
  Association}.
\newblock
\urldef\tempurl%
\url{http://static.usenix.org/events/sec11/tech/full\_papers/Haeberlen.pdf}
\showURL{%
\tempurl}


\bibitem[Hardt and Talwar(2010)]%
        {HardtT10}
\bibfield{author}{\bibinfo{person}{Moritz Hardt} {and} \bibinfo{person}{Kunal
  Talwar}.} \bibinfo{year}{2010}\natexlab{}.
\newblock \showarticletitle{On the geometry of differential privacy}. In
  \bibinfo{booktitle}{\emph{Proceedings of the 42nd {ACM} Symposium on Theory
  of Computing, {STOC} 2010, Cambridge, Massachusetts, USA, 5-8 June 2010}},
  \bibfield{editor}{\bibinfo{person}{Leonard~J. Schulman}} (Ed.).
  \bibinfo{publisher}{{ACM}}, \bibinfo{pages}{705--714}.
\newblock
\urldef\tempurl%
\url{https://doi.org/10.1145/1806689.1806786}
\showDOI{\tempurl}


\bibitem[june wunder et~al\mbox{.}(2022)]%
        {DBLP:journals/corr/abs-2202-01901}
\bibfield{author}{\bibinfo{person}{june wunder}, \bibinfo{person}{Arthur
  {Azevedo de Amorim}}, \bibinfo{person}{Patrick Baillot}, {and}
  \bibinfo{person}{Marco Gaboardi}.} \bibinfo{year}{2022}\natexlab{}.
\newblock \showarticletitle{Bunched Fuzz: Sensitivity for Vector Metrics}.
\newblock \bibinfo{journal}{\emph{CoRR}}  \bibinfo{volume}{abs/2202.01901}
  (\bibinfo{year}{2022}).
\newblock
\showeprint[arXiv]{2202.01901}
\urldef\tempurl%
\url{https://arxiv.org/abs/2202.01901}
\showURL{%
\tempurl}


\bibitem[McSherry and Talwar(2007)]%
        {McSherryT07}
\bibfield{author}{\bibinfo{person}{Frank McSherry} {and} \bibinfo{person}{Kunal
  Talwar}.} \bibinfo{year}{2007}\natexlab{}.
\newblock \showarticletitle{Mechanism Design via Differential Privacy}. In
  \bibinfo{booktitle}{\emph{48th Annual {IEEE} Symposium on Foundations of
  Computer Science {(FOCS} 2007), October 20-23, 2007, Providence, RI, USA,
  Proceedings}}. \bibinfo{publisher}{{IEEE} Computer Society},
  \bibinfo{pages}{94--103}.
\newblock
\urldef\tempurl%
\url{https://doi.org/10.1109/FOCS.2007.41}
\showDOI{\tempurl}


\bibitem[Moot and Retor{\'{e}}(2012)]%
        {DBLP:series/lncs/6850}
\bibfield{author}{\bibinfo{person}{Richard Moot} {and}
  \bibinfo{person}{Christian Retor{\'{e}}}.} \bibinfo{year}{2012}\natexlab{}.
\newblock \bibinfo{booktitle}{\emph{The Logic of Categorial Grammars - {A}
  Deductive Account of Natural Language Syntax and Semantics}}.
  \bibinfo{series}{Lecture Notes in Computer Science},
  Vol.~\bibinfo{volume}{6850}.
\newblock \bibinfo{publisher}{Springer}.
\newblock
\showISBNx{978-3-642-31554-1}
\urldef\tempurl%
\url{https://doi.org/10.1007/978-3-642-31555-8}
\showDOI{\tempurl}


\bibitem[Near et~al\mbox{.}(2019)]%
        {DBLP:journals/pacmpl/NearDASGWSZSSS19}
\bibfield{author}{\bibinfo{person}{Joseph~P. Near}, \bibinfo{person}{David
  Darais}, \bibinfo{person}{Chike Abuah}, \bibinfo{person}{Tim Stevens},
  \bibinfo{person}{Pranav Gaddamadugu}, \bibinfo{person}{Lun Wang},
  \bibinfo{person}{Neel Somani}, \bibinfo{person}{Mu Zhang},
  \bibinfo{person}{Nikhil Sharma}, \bibinfo{person}{Alex Shan}, {and}
  \bibinfo{person}{Dawn Song}.} \bibinfo{year}{2019}\natexlab{}.
\newblock \showarticletitle{Duet: an expressive higher-order language and
  linear type system for statically enforcing differential privacy}.
\newblock \bibinfo{journal}{\emph{Proc. {ACM} Program. Lang.}}
  \bibinfo{volume}{3}, \bibinfo{number}{{OOPSLA}} (\bibinfo{year}{2019}).
\newblock
\urldef\tempurl%
\url{https://doi.org/10.1145/3360598}
\showDOI{\tempurl}


\bibitem[O'Hearn(2003)]%
        {DBLP:journals/jfp/OHearn03}
\bibfield{author}{\bibinfo{person}{Peter~W. O'Hearn}.}
  \bibinfo{year}{2003}\natexlab{}.
\newblock \showarticletitle{On bunched typing}.
\newblock \bibinfo{journal}{\emph{J. Funct. Program.}} \bibinfo{volume}{13},
  \bibinfo{number}{4} (\bibinfo{year}{2003}), \bibinfo{pages}{747--796}.
\newblock
\urldef\tempurl%
\url{https://doi.org/10.1017/S0956796802004495}
\showDOI{\tempurl}


\bibitem[O'Hearn and Pym(1999)]%
        {DBLP:journals/bsl/OHearnP99}
\bibfield{author}{\bibinfo{person}{Peter~W. O'Hearn} {and}
  \bibinfo{person}{David~J. Pym}.} \bibinfo{year}{1999}\natexlab{}.
\newblock \showarticletitle{The logic of bunched implications}.
\newblock \bibinfo{journal}{\emph{Bull. Symb. Log.}} \bibinfo{volume}{5},
  \bibinfo{number}{2} (\bibinfo{year}{1999}).
\newblock
\urldef\tempurl%
\url{https://doi.org/10.2307/421090}
\showDOI{\tempurl}


\bibitem[Reed and Pierce(2010)]%
        {DBLP:conf/icfp/ReedP10}
\bibfield{author}{\bibinfo{person}{Jason Reed} {and}
  \bibinfo{person}{Benjamin~C. Pierce}.} \bibinfo{year}{2010}\natexlab{}.
\newblock \showarticletitle{Distance makes the types grow stronger: a calculus
  for differential privacy}. In \bibinfo{booktitle}{\emph{{ICFP} 2010}}.
  \bibinfo{publisher}{{ACM}}.
\newblock
\urldef\tempurl%
\url{https://doi.org/10.1145/1863543.1863568}
\showDOI{\tempurl}


\bibitem[Toro et~al\mbox{.}(2020)]%
        {DBLP:journals/corr/abs-2010-11342}
\bibfield{author}{\bibinfo{person}{Mat{\'{\i}}as Toro}, \bibinfo{person}{David
  Darais}, \bibinfo{person}{Chike Abuah}, \bibinfo{person}{Joe Near},
  \bibinfo{person}{Federico Olmedo}, {and} \bibinfo{person}{{\'{E}}ric
  Tanter}.} \bibinfo{year}{2020}\natexlab{}.
\newblock \showarticletitle{Contextual Linear Types for Differential Privacy}.
\newblock \bibinfo{journal}{\emph{CoRR}}  \bibinfo{volume}{abs/2010.11342}
  (\bibinfo{year}{2020}).
\newblock
\showeprint[arXiv]{2010.11342}
\urldef\tempurl%
\url{https://arxiv.org/abs/2010.11342}
\showURL{%
\tempurl}


\bibitem[Winograd{-}Cort et~al\mbox{.}(2017)]%
        {DBLP:journals/pacmpl/Winograd-CortHR17}
\bibfield{author}{\bibinfo{person}{Daniel Winograd{-}Cort},
  \bibinfo{person}{Andreas Haeberlen}, \bibinfo{person}{Aaron Roth}, {and}
  \bibinfo{person}{Benjamin~C. Pierce}.} \bibinfo{year}{2017}\natexlab{}.
\newblock \showarticletitle{A framework for adaptive differential privacy}.
\newblock \bibinfo{journal}{\emph{Proc. {ACM} Program. Lang.}}
  \bibinfo{volume}{1}, \bibinfo{number}{{ICFP}} (\bibinfo{year}{2017}),
  \bibinfo{pages}{10:1--10:29}.
\newblock
\urldef\tempurl%
\url{https://doi.org/10.1145/3110254}
\showDOI{\tempurl}


\end{thebibliography}
